\documentclass[10pt]{article}
\usepackage[margin=1in]{geometry}
\usepackage{amsfonts,amssymb}
\usepackage{amsthm}
\usepackage{hyperref}

\usepackage{bigfoot}
\usepackage[utf8]{inputenc}

\usepackage{amsmath}
\usepackage[nameinlink]{cleveref}
\crefname{theorem}{thm.}{thm.}
\crefname{section}{sec.}{sec.}
\crefname{corollary}{cor.}{cor.}
\crefname{lemma}{lem.}{lem.}
\Crefname{theorem}{Thm.}{Thm.}
\Crefname{section}{Sec.}{Sec.}
\Crefname{figure}{Fig.}{Fig.}
\Crefname{equation}{Eqn.}{Eqn.}
\Crefname{assumption}{Assumption}{Assumption}
\usepackage{fancyhdr}
\usepackage{adjustbox}
\usepackage{graphicx}
\usepackage{tikz}
\usepackage{tikzscale}
\usepackage{quantikz}
\usepackage{pgfplots}
\usepackage{pgf-pie}
\pgfplotsset{compat=1.15}
\usepackage{xcolor}
\definecolor{classical}{HTML}{DBD3F5}
\definecolor{quantum}{HTML}{FBE7E6}
\usetikzlibrary{arrows.meta,decorations.markings}
\tikzset{
     heat/.style={brown,thick,-Latex, decoration={coil,aspect=0},decorate}
  }

\usepackage{wrapfig}
\usepackage{lscape}
\usepackage{rotating}

\usepackage{algorithm}
\usepackage{algcompatible}
\usepackage[utf8]{inputenc}
\usepackage{amsthm}

\usepackage{mathrsfs}
\usepackage[normalem]{ulem}
\usepackage{array}
\usepackage{subcaption}

\usepackage{multicol}
\usepackage{cuted} 
\usepackage{booktabs}
\usepackage{fancyhdr}

\usepackage{thm-restate}

\newtheorem{theorem}{Theorem}
\newtheorem*{theorem*}{Theorem}
\newtheorem{lemma}{Lemma}
\newtheorem{assumption}{Assumption}

\newtheorem{remark}{Remark}

\allowdisplaybreaks

\newcommand{\bv}{\mathbf{b}}
\newcommand{\xv}{\mathbf{x}}
\newcommand{\yv}{\mathbf{y}}
\newcommand{\psiv}{\mathbf{\psi}}
\newcommand{\phiv}{\mathbf{\phi}}
\newcommand{\pv}{\mathbf{p}}
\newcommand{\Zev}{\mathbf{0}}
\newcommand{\bvt}{\tilde{\mathbf{b}}}

\newcommand{\Hm}{\mathbf{H}}
\newcommand{\Am}{\mathbf{A}}
\newcommand{\Amt}{\tilde{\mathbf{A}}}
\newcommand{\Km}{\mathbf{K}}
\newcommand{\Phim}{\mathbf{\Phi}}
\newcommand{\Xm}{\mathbf{X}}
\newcommand{\Ym}{\mathbf{Y}}
\newcommand{\Zm}{\mathbf{Z}}
\newcommand{\Um}{\mathbf{U}}

\newcommand{\In}{\mathbf{I}}
\newcommand{\Id}{\mathbf{I}}
\newcommand{\Pm}{\mathbf{P}}
\newcommand{\Vm}{\mathbf{V}}
\newcommand{\Sp}{\sigma_+}
\newcommand{\Sm}{\sigma_-}

\newcommand{\thetav}{\mathbf{\theta}}
\newcommand{\ra}{\rangle}
\newcommand{\la}{\langle}

\newcommand{\Ar}{\mathbf{A}}
\newcommand{\uv}{\mathbf{u}}
\newcommand{\uvh}{\hat{\mathbf{u}}}

\newcommand{\nx}{n} 
\newcommand{\nc}{n_c} 
\newcommand{\np}{n_p} 
\newcommand{\nt}{N} 
\newcommand{\tsteps}{M} 

\newcommand{\Rr}{\mathbb{R}}
\newcommand{\Cr}{\mathbb{C}}
\newcommand{\Nr}{\mathbb{N}}
\newcommand{\tra}{T}
\newcommand{\tauv}{\mathbf{\tau}}

\newcommand{\Dm}{\mathbf{D}}

\newcommand{\etav}{\mathbf{\eta}}
\newcommand{\wv}{\mathbf{w}}
\newcommand{\wvh}{\hat{\mathbf{w}}}
\newcommand{\wvo}{\overline{\mathbf{w}}}
\newcommand{\wvc}{\mathbf{w}}

\newcommand{\hv}{\mathbf{h}}

\newcommand{\F}{\mathbf{F}}
\newcommand{\tF}{\tilde{\mathbf{F}}}
\newcommand{\lam}{\lambda}

\newcommand{\nxx}{n_x}
\newcommand{\tstepst}{n_t}

\usepackage{listings}

\definecolor{backcolour}{rgb}{0.95,0.95,0.92}
\definecolor{codegreen}{rgb}{0,0.6,0}

\lstdefinestyle{myStyle}{
    backgroundcolor=\color{backcolour},   
    commentstyle=\color{codegreen},
    basicstyle=\ttfamily\normalsize,
    breakatwhitespace=false,         
    breaklines=true,                 
    keepspaces=true,    
    numbersep=5pt,                  
    showspaces=false,                
    showstringspaces=false,
    showtabs=false,                  
    tabsize=2,
}

\usepackage{color}
\definecolor{deepblue}{rgb}{0,0,0.5}
\definecolor{deepred}{rgb}{0.6,0,0}
\definecolor{deepgreen}{rgb}{0,0.5,0}

\newcommand\pythonstyle{\lstset{
language=Python,
basicstyle=\ttfamily,
morekeywords={self},              
keywordstyle=\ttfamily\color{deepblue},
emph={MyClass,__init__},          
emphstyle=\ttfamily\color{deepred},    
stringstyle=\color{deepgreen},
frame=tb,                         
showstringspaces=false
}}

\usetikzlibrary{calc,trees,positioning,arrows,chains,shapes.geometric,%
    decorations.pathreplacing,decorations.pathmorphing,shapes,%
    matrix,shapes.symbols}

\usetikzlibrary{positioning,fit,shapes.geometric,backgrounds}
    
\tikzset{
>=stealth',
  punktchain/.style={
    rectangle, 
    rounded corners, 
    draw=black, very thick,
    text width=11em, 
    minimum height=8em, 
    text centered, 
    on chain},
   chain/.style={
    rectangle, 
    rounded corners, 
    draw=black, very thick,
    text width=11em, 
    minimum height=8em, 
    text centered 
    },
  line/.style={draw, thick, <-},
   punktchainVert/.style={
    rectangle, 
    rounded corners, 
    draw=black, very thick,
    text width=20em, 
    minimum height=8em, 
    text centered, 
    on chain},
  line/.style={draw, thick, <-}
  element/.style={
    tape,
    top color=white,
    bottom color=blue!50!black!60!,
    minimum width=8em,
    draw=blue!40!black!90, very thick,
    text width=10em, 
    minimum height=3.5em, 
    text centered, 
    on chain},
  every join/.style={->, thick,shorten >=1pt},
  decoration={brace},
  tuborg/.style={decorate},
  tubnode/.style={midway, right=2pt},
}

\usetikzlibrary{
  arrows.meta,                        
  backgrounds,                        
  ext.paths.ortho,                    
  ext.positioning-plus,               
  ext.node-families.shapes.geometric, 
  calc}                               
\tikzset{
  basic box/.style={
    shape=rectangle, rounded corners, align=center, draw=#1, fill=#1!25},
  header node/.style={
    node family/width=header nodes,
    text depth=+.3ex, fill=white, draw},
  header/.style={%
    inner ysep=+1.5em,
    append after command={
      \pgfextra{\let\TikZlastnode\tikzlastnode}
      node [header node] (header-\TikZlastnode) at (\TikZlastnode.north) {#1}
      node [span=(\TikZlastnode)(header-\TikZlastnode)]
           at (fit bounding box) (h-\TikZlastnode) {}
    }
  },
  fat blue line/.style={ultra thick, blue},
  fat black line/.style={thick, black}
} 

\lstnewenvironment{python}[1][]
{
\pythonstyle
\lstset{#1}
}
{}

\usepackage{appendix}

\DeclareNewFootnote{AAffil}[arabic]
\DeclareNewFootnote{ANote}[fnsymbol]

\usepackage{etoolbox}
\makeatletter
\patchcmd\maketitle{\def\@makefnmark{\rlap{\@textsuperscript{\normalfont\@thefnmark}}}}{}{}{}
\makeatother

\makeatletter
\def\thanksAAffil#1{
  \footnotemarkAAffil\protected@xdef\@thanks{\@thanks%
        \protect\footnotetextAAffil[\the \c@footnoteAAffil]{#1}}%
}
\def\thanksANote#1{%
  \footnotemarkANote%
  \protected@xdef\@thanks{\@thanks%
        \protect\footnotetextANote[\the \c@footnoteANote]{#1}}%
}
\makeatother

\author{
Abeynaya Gnanasekaran%
\thanks{Equal contribution}$^{\hspace{4pt}, }$\thanks{abeynaya.gnanasekaran@rtx.com, SRI International, Menlo Park, CA, USA.}%
  \hspace{5pt}, %
  Amit Surana%
  \footnotemark[1]$^{\hspace{4pt}, }$\thanks{amit.surana@rtx.com, RTX Technology Research Center (RTRC), East Hartford, CT, USA.}%
  \hspace{5pt}, %
  Hongyu Zhu \thanks{RTX Technology Research Center (RTRC), East Hartford, CT, USA.}%
}

\title{Variational Quantum Framework for Nonlinear PDE Constrained Optimization Using Carleman Linearization}

\begin{document}
\maketitle
\abstract{We present a novel variational quantum framework for nonlinear partial differential equation (PDE) constrained optimization problems.
The proposed work extends the recently introduced bi-level variational quantum PDE constrained optimization (BVQPCO) framework for linear PDE to a nonlinear setting by leveraging Carleman linearization (CL). CL framework allows one to transform a system of polynomial ordinary differential equations (ODE), i,e. ODE with polynomial vector field,  into an system of infinite but linear system of ODE. For instance, such polynomial ODEs naturally arise when the PDE are semi-discretized in the spatial dimensions. By truncating the CL system to a finite order, one obtains a finite system of linear ODE to which the linear BVQPCO framework can be applied. In particular, the finite system of linear ODE is discretized in time and embedded as a system of linear equations. The variational quantum linear solver (VQLS) is used to solve the linear system for given optimization parameters, and evaluate the design cost/objective function, and a classical black box optimizer is used to select next set of parameter values based on this evaluated cost. We present detailed computational error and complexity analysis and prove that under suitable assumptions, our proposed framework can provide potential advantage over classical techniques. We implement our framework using the PennyLane library and apply it to solve inverse Burgers' problem. We also explore an alternative tensor product decomposition which exploits the sparsity/structure of linear system arising from PDE discretization to facilitate the computation of VQLS cost functions. 

\section{Introduction}
We present a variational quantum framework for optimization problems constrained by nonlinear partial differential equations (PDE). This extends the recently proposed bi-level variational quantum PDE constrained optimization (BVQPCO) framework  for linear PDEs in \cite{surana2024variational} to a nonlinear setting. This is an important extension as many science and engineering applications such as aerodynamics, computational fluid dynamics (CFD), material science, computer vision, and inverse problems necessitate the solution of optimization problems constrained by a system of nonlinear PDE. Examples include Euler and Navier-Stokes PDEs and heat equation in fluid dynamics, combustion and weather forecasting; magnetohydrodynamics and Vlasov-Maxwell PDEs in plasma dynamics and astrophysics; wave equation in structural mechanics; Black–Scholes PDE in finance; and reaction-diffusion PDEs in chemistry, biology and ecology to name a few.  Closed-form solutions are generally unavailable for PDE constrained optimization problems, necessitating the development of numerical methods~\cite{biegler2003large,hinze2008optimization,antil2018frontiers}. A variety of classical gradient-based and gradient-free numerical methods have been developed, and rely on repeated calls to the underlying PDE solver. Since PDE simulations are computationally expensive, using them within the design/optimization loop can become a bottleneck. 

To extend the BVQPCO framework for nonlinear PDE we employ Carleman linearization (CL) framework. The CL framework allows one to transform polynomial ordinary differential equations (ODE), i.e., ODE with polynomial vector field, into an system of infinite but linear ODE~\cite{kowalski1991nonlinear}. For instance, such polynomial ODE naturally arise when PDE, e.g., as mentioned above, are semi-discretized in spatial dimensions. By truncating the CL system to a finite order, one obtains a finite system of linear ODE and under suitable conditions the impact of truncation error can be characterized \cite{forets2017explicit,amini2022carleman}. CL has become an important tool in developing quantum algorithms for simulating polynomial ODE on quantum devices since the seminal work \cite{liu2021efficient}. Specifically, assuming $R_2<1$, where $R_2$ is a parameter characterizing the ratio of the nonlinearity and forcing to the linear dissipation, the query/gate complexity of proposed algorithm takes the form $\frac{sT^2 q}{\epsilon}\textsf{poly}(\log T, \log \nx, \log 1/\epsilon)$, where, $\epsilon$ is desired solution accuracy, $\nx$ is system dimension,  $q=\|\xv(0)\|/\|\xv(T)\|$ measures decay of the solution $\xv(t)\in \Rr^{\nx}$, $T$ is the period of integration,  and $s$ is sparseness of system matrices describing the ODE. This algorithm/analysis was extended for systems of  $k$-th degree polynomial ODE for arbitrary (finite) values of $k$ in \cite{surana2024efficient}. Recent work has shown that by transforming the problem or working with a different formulation of conservation laws in CFD applications, e.g. lattice Boltzman method (LBM), one may be able to more readily satisfy such conditions \cite{li2023potential}. Furthermore, \cite{penuel2024feasibility} developed detailed quantum resource estimates for implementing CL based LBM simulation of incompressible CFD on a fully fault tolerant quantum computer. While the motivation of exploring quantum algorithms in \cite{penuel2024feasibility} is to accelerate simulation based design which results in a PDE constrained optimization, there was no algorithmic treatment of that problem. The application of CL for solving polynomial ODE in a variational quantum setting has also been numerically investigated in \cite{demirdjian2022variational}, in which QLSA is replaced by the variational quantum linear solver (VQLS) \cite{VQLS}.  This work was also restricted to applying CL+VQLS for PDE simulation and did not consider any associated design/optimization problem. Furthermore, the work does not make any computational error/complexity assessments for solving polynomial ODE by the CL+VQLS approach. 

The main contributions of this work are as follows:
\begin{itemize}
  \item We formulate a bi-level variational quantum framework, referred to as nonlinear BVQPCO (nBVQPCO),  for solving nonlinear PDE constrained optimization problems as discussed above, 
   
  \item We present a detailed computational error analysis for CL+VQLS based solution of polynomial ODE and prove that under suitable assumptions the CL+VQLS approach can generate a normalized solution of given ODE to an arbitrary accuracy. Furthermore by combining this error analysis with the empirically known results for query complexity of VQLS, we assess potential utility of our nBVQPCO framework over equivalent classical methods.  
  
  \item We implement the nBVQPCO framework using the PennyLane library and apply it to solve a prototypical inverse problem applied to nonlinear Burgers' PDE which involves calibrating the PDE model parameter to match the given measurements. We also explore an alternative tensor product decomposition which exploits the sparsity/structure of linear system arising from PDE discretization to facilitate the computation of VQLS cost functions.
\end{itemize}

The paper is organized as follows. We start with mathematical preliminaries in~\Cref{sec:not} and introduce the nonlinear PDE constrained optimization problem in~\Cref{sec:PDEopt}. We describe the CL+VQLS procedure in~\Cref{sec:QCquadODE} and develop the nBVQPCO framework extension including an outline of pseudo-code for its implementation in~\Cref{sec:nBVQPCO}. Complexity and error analysis of the proposed framework is covered in~\Cref{sec:erranalysis}. In ~\Cref{sec:num} we formulate the inverse PDE problem in our nBVQPCO framework, apply an alternative tensor product decomposition for computing VQLS cost function, and describe the implementation details and numerical results using the PennyLane library. We conclude in~\Cref{sec:conc} with avenues for future work.  

\section{Mathematical Preliminaries}\label{sec:not}
Let $\Nr = \{1, 2,\dots\}$, $\Rr$, and $\Cr$ be the sets of positive integers,  real numbers, and  complex numbers respectively. We will denote vectors by small bold letters and  matrices by capital bold letters. $\Am^{\dagger}$ and $\Am^{\prime}$ will denote the vector/matrix complex conjugate and vector/matrix transpose, respectively. $Tr(\Am)$ will denote the trace of a matrix. We will represent an identity matrix of size $s \times s$ by $\Id_s$.  

Kronecker product will be denoted by $\otimes$. For any pair of vectors, $\xv\in \Rr^n$ and $\yv \in \Rr^m$, their Kronecker product $\wv\in\Rr^{nm}$ is
\begin{equation}\label{eq:kron}
\wv=\xv\otimes\yv= (x_1y_1, x_1y_2, . . . , x_1y_m, x_2y_1, . . . , x_2y_m, . . . , x_ny_1, . . . , x_ny_m)^\tra.
\end{equation}
Similarly, for $\Am\in\Rr^{m\times n}$ and $\mathbf{B}\in\Rr^{p\times q}$, their Kronecker product is defined as
\begin{equation}\label{eq:kronM}
\mathbf{C}=\Am\otimes\mathbf{B}=\left(
                                 \begin{array}{ccc}
                                   a_{11}\mathbf{B} & \cdots & a_{1n}\mathbf{B} \\
                                   \vdots & \vdots & \vdots \\
                                    a_{m1}\mathbf{B} & \cdots & a_{mn}\mathbf{B} \\
                                 \end{array}
                               \right),
\end{equation}
where, $\mathbf{C}\in \Rr^{mp\times nq}$.  The Kronecker power is a convenient notation to express all possible products of elements
of a vector up to a given order, and is denoted by
\begin{equation}\label{eq:powx}
\xv^{[i]}=\underbrace{\xv\otimes\xv\cdots\otimes\xv}_{i-\textsf{times}},
\end{equation}
for any $i\in\Nr$ with the convention $\xv^{[0]}=1$. Moreover, $\mbox{dim} (\xv^{[i]})= n^i$,  and each component of $\xv^{[i]}$ is of the form $x_1^{\omega_1}x_2^{\omega_2}\cdots x_n^{\omega_n} $ for some multi-index $\mathbf{\omega}\in \Nr^n$  of weight $\sum_{j=1}^n\omega_j=i$.
Similarly, we denote the matrix Kronecker power by,
\begin{equation}\label{eq:powA}
\Ar^{[i]}=\underbrace{\Ar\otimes\Ar\dots\otimes\Ar}_{i-\textsf{times}}.
\end{equation}

The standard inner product between two complex vectors $\mathbf{x},\mathbf{y}\in \Rr^n$ will be denoted by $\la\mathbf{x},\mathbf{y}\ra$. The $l_p$-norm in the Euclidean space $\Rr^n$ will be denoted by $\|\cdot\|_p,p=1,2,\dots,\infty$ and is defined as follows
\begin{equation}\label{eq:normx}
\|\xv\|_p=\left(\sum_{j=1}^n|x_j|^p\right)^{1/p}.
\end{equation}
For the norm of a matrix $\Am \in \Rr^{n\times m}$, we use the induced norm, namely
\begin{equation}\label{eq:normA}
\|\Am\|_p=\max_{\xv \in \Rr^m} \frac{\|\Am \xv\|_p}{\|\xv\|_p}.
\end{equation}
We will use $p=2$, i.e., $l_2$ norm for vectors and spectral norm for matrices unless stated otherwise. The spectral radius of a matrix denoted by $\rho(\Am)$ is defined as 
\begin{equation}\label{eq:specA}
\rho(\Am)=\max\{|\lambda|:\lambda \mbox{ eignenvalues of } \Am\}.
\end{equation}
Trace norm of a matrix $\Am$ is defined as $\|\Am\|_{Tr}=Tr(\sqrt{\Am^*\Am})$ which is the  Schatten $q$-norms with $q=1$.

Let $\mathcal{F}$ be a vector space of real vector valued functions $\uv(t):[0,T]\rightarrow \Rr^n$ defined on $[0,T]$. Then the inner product between $\uv_1,\uv_2\in \mathcal{F}$ is defined as
\begin{equation}\label{eq:funcip}
\la \uv_1,\uv_2\ra_T=\int_{0}^{T}\la \uv_1(t),\uv_2(t)\ra dt.
\end{equation}

We will use standard braket notation, i.e. $|\psiv\ra$ and $\la\psiv|$ in representing the quantum states \cite{nielsen2010quantum}. The inner product between two quantum states $|\psiv\ra$ and $|\phiv\ra$ will be denoted by $\la\psiv|\phiv\ra$. For a vector $\mathbf{x}$, we denote by $|\mathbf{x}\ra=\frac{\mathbf{x}}{\|\mathbf{x}\|}$ as the corresponding quantum state. The trace norm $\rho(\psiv,\phiv)$ between two pure states $|\psiv\ra$ and $|\phiv\ra$ is defined as
\begin{equation}\label{eq:trace}
\rho(\psiv,\phiv)=\frac{1}{2}\| |\phiv\ra \la\phiv| -|\psiv\ra \la\psiv| \|_{Tr}=\sqrt{1-|\la \phiv|\psiv \ra|^2}.
\end{equation}

\section{Variational Quantum Framework for Nonlinear PDE Constrained Optimization}\label{sec:VQNPCO}

\subsection{Nonlinear PDE Constrained Optimization Problem}\label{sec:PDEopt}
We consider a general class of PDE constrained design optimization problems of the form~\cite{de2015numerical}
\begin{eqnarray}
&&\min_{\mathbf{p},\mathbf{u}} C_d(\mathbf{u}, \mathbf{p}) \label{eq: gen_opt}\\
\text{s.t.} &&\mathcal{F}(\mathbf{u},\mathbf{p},t) = 0,\label{eq:pdecons}\\
&&\mathcal{B}(\mathbf{u},\mathbf{p},t) = 0,\label{eq:bcic}\\
&& g_i(\mathbf{p})\leq 0, i=1,\cdots,\nc,
\end{eqnarray}
where $t\in[0,T]$ for given $T\in\Rr$, $\mathbf{p}\in\Rr^{\np}$ is the vector of design variables (e.g., material properties, aerodynamic shape), $C_d$ is the cost function (e.g., heat flux, drag/ lift) and $\mathcal{F}(\mathbf{u}, \mathbf{p},t)$ are the PDE constraints (e.g., conservation laws, constitutive relations) with $\mathbf{u}(\mathbf{x},t;\pv)$ being the solution of the PDE defined as a function of space $\mathbf{x}$ and time $t$ for given parameters $\pv$, $\mathcal{B}(\mathbf{u},\mathbf{p},t)$ capture the constraints coming from the PDE boundary and initial conditions, and $g_i(\mathbf{p}),i=1,\cdots,\nc$ are constraints on the parameters. We make the following assumptions.
\begin{assumption}\label{assump:poly}
When semi-discretized in space, the PDE (\ref{eq:pdecons}) and associated boundary/initial conditions (\ref{eq:bcic}) results in a polynomial ODE of the form
\begin{eqnarray}
\dot{\uv}(t)&=&\F_0(t,\pv)+\F_1(t,\pv)\uv+\F_2(t,\pv)\uv^{[2]}+\cdots+\F_k(t,\pv)\uv^{[k]},\label{eq:gen}\\
\uv(0)&=&\uv_0(\pv)\in\Rr^\nx,\notag
\end{eqnarray}
where, $k\in \Nr$ is the order of the polynomial,  $\uv(t)\in \Rr^\nx, t\in[0,T]$ (with slight about of notation) is an $\nx$-dimensional solution vector and $\F_i(t,\pv)\in\Rr^{\nx\times \nx^i}$ are in general time dependent matrices. We assume that the boundary conditions have already been accounted for in this ODE representation, see~\Cref{sec:invprob} for an example.
\end{assumption}
\begin{assumption}\label{assump:costfunc}
The cost function is a quadratic function of the solution $\uv(t), t\in[0,T]$, i.e.
\begin{equation}
C_d(\uv, \pv)=f(w_1\la\uv,\Hm\uv\ra_T+w_2\la\uv,\hv\ra_T,\pv),\label{eq:cstform}
\end{equation}
where $w_1,w_2\in \Rr$, $\Hm \in \Rr^{\nx \times \nx}$ and $\hv\in \Rr^{\nx}$.
\end{assumption}
Under these assumptions, the optimization problem (\ref{eq: gen_opt}) can be represented as
\begin{eqnarray}
&&\min_{\mathbf{p},\mathbf{u}} C_d(\mathbf{u}, \mathbf{p})\equiv f(w_1\la\uv,\Hm\uv\ra_T+w_2\la\uv,\hv\ra_T,\pv) \label{eq:gen_optrefor}\\
\text{s.t.} && \dot{\uv}(t)=\F_0(t,\pv)+\F_1(t,\pv)\uv+\cdots+\F_k(t,\pv)\uv^{[k]},\label{eq:odecons}\\
&& \uv(0)=\uv_0(\pv),\label{eq:ic}\\
&& g_i(\mathbf{p})\leq 0, i=1,\cdots,\nc. \label{eq:pc}
\end{eqnarray}

\begin{remark}\label{rem:polyjust}
\Cref{assump:poly} is fairly general as several PDEs arising in engineering and science result in polynomial ODEs when semi-discretized in space, see \cite{surana2024efficient} and references therein. Examples include  Euler and Navier-Stokes PDEs and heat equation in fluid dynamics, combustion and weather forecasting; magnetohydrodynamics and Vlasov-Maxwell PDEs in plasma dynamics and astrophysics; wave equation in structural mechanics; Black–Scholes PDE in finance; and reaction-diffusion PDEs in chemistry, biology and ecology to name a few. Furthermore, polynomial ODEs also arise in mechanics, molecular dynamics, chemical kinetics, epidemiology, social dynamics, and biological networks. 
\end{remark}

\begin{remark}\label{rem:quadcost}
\Cref{assump:costfunc} commonly arises in many PDE constrained optimization problems. For instance in inverse problems the cost function is typically taken as a squared norm of difference between the solution and measured variables, which results in a quadratic form, see~\Cref{sec:invprob} for details. Other examples include PDE optimal control where the cost function is a weighted combination of deviation from a reference signal and squared norm of control input, aerodynamic design where the objective is to minimize heat transfer/drag or maximize lift, and computer vision problems such as shape from shading, surface reconstruction from sparse data and optical flow computation where the objective function depends on an appropriate form of squared norm.
\end{remark}

The BVQPCO framework proposed in \cite{surana2024variational} cannot be applied directly to the optimization problem (\ref{eq:gen_optrefor})-(\ref{eq:pc}), as the constraints are nonlinear. To address this we employ CL framework, whereby the polynomial ODEs are transformed into an infinite-dimensional system of linear ODEs and truncated. Since, it is always possible to map the $k$-th degree polynomial ODE (\ref{eq:gen}) to a higher dimensional quadratic polynomial ODE, and then apply CL \cite{surana2024efficient}, without loss of generality we restrict $k=2$ throughout this paper. 

\subsection{Solving Nonlinear ODEs via Carelman Linerization based VQLS}\label{sec:QCquadODE}
Consider a system of inhomogeneous quadratic polynomial ODE (\ref{eq:odecons}) with $k=2$, i.e.
\begin{eqnarray}
\dot{\uv}&=&\F_0(t,\pv)+\F_1(t,\pv)\uv+\F_2(t,\pv)\uv^{[2]}, \label{eq:qdcs}\\
\uv(0)&=&\uv_0\in\Rr^\nx.\notag
\end{eqnarray}
Below we describe steps required to transform solution of these nonlinear ODEs into a form amenable for VQLS, see Fig.~\ref{fig:workflow} for illustration of steps involved. 

\paragraph{Step 1:} For CL we introduce variables $\wv_i=\uv^{[i]}\in \Rr^{\nx^i}, i\in \Nr$, which satisfy
\begin{equation}\label{eq:genz}
\dot{\wv}_i=\sum_{j=0}^{2}\Ar^{i}_{i+j-1}(t,\pv)\wv_{i+j-1},
\end{equation}
where, $\Ar^{i}_{i+j-1}\in\Rr^{\nx^i\times \nx^{i+j-1}}$ is given by
\begin{equation}\label{eq:Adef}
\Ar^{i}_{i+j-1}(t,\pv)=\sum_{p=1}^{i}\overbrace{\Id_{\nx\times \nx}\otimes\cdots\otimes\underbrace{\F_j(t,\pv)}_{p-\textsf{th position}}\otimes\cdots\otimes \Id_{\nx\times \nx}}^{i \textsf{ factors}},
\end{equation}
with $0 \leq j\leq 2$. Note that (\ref{eq:genz}) defines an infinite dimensional system of linear ODE with state $\wv_{\infty}(t)=(\wv_1^\tra(t),\wv_2^\tra(t),\cdots)^\tra $ and initial condition $\wv_{\infty}(0)=(\uv_0^\tra,(\uv^{[2]}_0)^\tra,\cdots)^\tra$, and is known as CL. For any finite $\nt\in \Nr$ define 
\begin{equation}\label{eq:wup}
\wvc(t)=\left(\begin{array}{c}
                                          \wv_1(t) \\
                                          \wv_2(t) \\
                                          \wv_3(t)\\
                                          \vdots \\
                                          \wv_{\nt}(t)
                                        \end{array}\right)\in \Rr^{\nt_c},
\end{equation}
where $\nt_c=\frac{\nx^{\nt+1}-\nx}{\nx-1}$. 
The CL when truncated to order $\nt$, results in a finite system of linear ODEs
\begin{eqnarray}
\dot{\wvh}&=&\Ar_\nt(t,\pv)\wvh+\bv(t,\pv),\label{eq:carlfq}\\
\wvh(0)&=&(\uv_0^\tra,(\uv^{[2]}_0)^\tra,\cdots,(\uv^{[N]}_0)^\tra)^\tra,
\end{eqnarray}
where
\begin{equation}\label{eq:AN2}
\Ar_\nt(t,\pv)=\left(\begin{array}{ccccc}
                \Ar_1^1(t,\pv) & \Ar^1_2(t,\pv) & 0 & \cdots & 0  \\
                \Ar^2_1(t,\pv) & \Ar_2^2(t,\pv) & \Ar^2_3(t,\pv) & \cdots & 0 \\
                0 & \Ar^3_2(t,\pv) & \Ar_3^3(t,\pv) & \Ar^3_4 (t,\pv)& 0  \\
                \vdots & \vdots & \vdots & \vdots & \vdots  \\
                0 & 0 & 0 & \Ar^{N}_{N-1}(t,\pv) & \Ar^{N}_N(t,\pv)\\
              \end{array}
            \right),
\end{equation}
\begin{equation}\label{eq:what}
\wvh(t)=\left(\begin{array}{c}
                                          \wvh_1(t) \\
                                          \wvh_2(t) \\
                                          \wvh_3(t)\\
                                          \vdots \\
                                          \wvh_{\nt}(t)
                                        \end{array}\right)\in \Rr^{\nt_c}, \quad \bv(t,\pv)=\left(\begin{array}{c}
                                          \F_0(t,\pv) \\
                                          0 \\
                                          0\\
                                          \vdots \\
                                          0
                                        \end{array}\right)
\end{equation}
Note that $\wvh(t)$ approximates $\wvc(t)$ as defined in~\Cref{eq:wup} 
and the error incurred as a function of the truncation level $\nt$ is described by the Lemma \ref{lem:cltruc}.

%

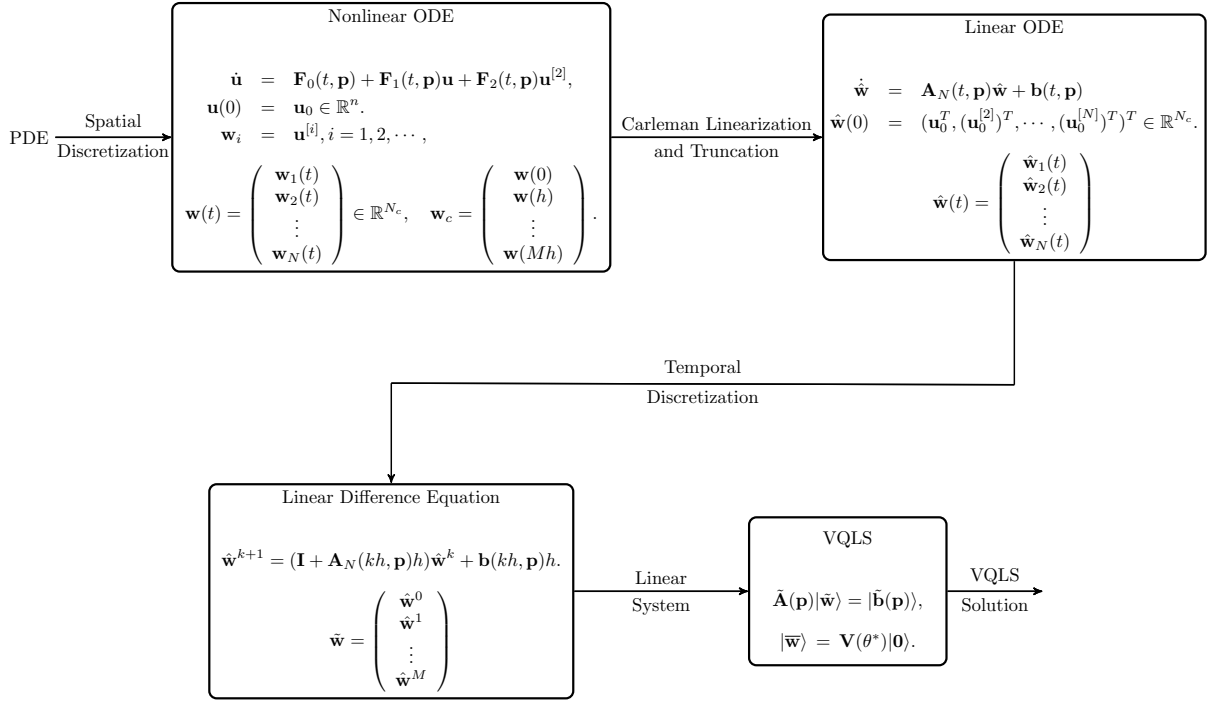
\begin{figure}
\centering
\scalebox{0.7}{
\begin{tikzpicture}
[node distance=4cm,]
\node[chain,text width= 23em] (intro){Nonlinear ODE\\ 
\begin{eqnarray*}
\dot{\uv}&=&\F_0(t,\pv)+\F_1(t,\pv)\uv+\F_2(t,\pv)\uv^{[2]},\\
\uv(0)&=&\uv_0\in\Rr^\nx.\notag\\
\wv_i&=&\uv^{[i]}, i=1,2,\cdots,
\end{eqnarray*}
\begin{equation*}
\wvc(t)=\left(\begin{array}{c}
                                          \wv_1(t) \\
                                          \wv_2(t) \\
                                          \vdots \\
                                          \wv_{\nt}(t)
                                        \end{array}\right)\in \Rr^{\nt_c},\quad \wv_c=\left(
        \begin{array}{c}
         \wvc(0) \\
         \wvc(h) \\
          \vdots \\
          \wvc(\tsteps h) \\
        \end{array}
      \right).
\end{equation*}
};

\node[chain,text width= 20em] 
(carleman) [right=of intro] {Linear ODE\\ 
\begin{eqnarray*}
\dot{\wvh}&=&\Ar_\nt(t,\pv)\wvh+\bv(t,\pv)\\
\wvh(0)&=&(\uv_0^\tra,(\uv^{[2]}_0)^\tra,\cdots,(\uv^{[N]}_0)^\tra)^\tra\in \Rr^{\nt_c}.
\end{eqnarray*}
\begin{equation*}
\wvh(t)=\left(\begin{array}{c}
                                          \wvh_1(t) \\
                                          \wvh_2(t) \\
                                          \vdots \\
                                          \wvh_{\nt}(t)
                                        \end{array}\right)
\end{equation*}
};

\node[chain, text width= 19em] (euler) [below=of intro] {Linear Difference Equation\\ 
\begin{equation*}
\wvh^{k+1}=(\Id+\Ar_N(kh,\pv)h)\wvh^k+\bv(kh,\pv)h.
\end{equation*}
\begin{equation*}
\tilde{\wv}=\left( \begin{array}{c}
         \wvh^0 \\
         \wvh^1 \\
         \vdots \\
         \wvh^\tsteps \\
       \end{array} \right)
\end{equation*}
};
\node[chain, text width= 10em, shift=(left:2em)] (vqls)  [right=of euler]    {VQLS\\ 
\begin{equation*}
\tilde{\Am}(\pv)|\tilde{\wv}\ra=|\tilde{\bv}(\pv)\ra,
\end{equation*}
$|\wvo\ra=\Vm(\thetav^*)|\Zev\ra$.};

\node (foo) [left=2.2cm of intro] {PDE};
\node (foo1) [below=2.1cm of carleman] {};
\node (foo2) [above=1.9cm of euler] {};
\node (end) [right=1.8cm of vqls] {};

\draw [thick, ->]  (foo)  to node[above] {Spatial} (intro);
\draw [thick, ->]  (foo)  to node[below] {Discretization} (intro);
\draw [thick, ->] (intro) to node[above] {Carleman Linearization} (carleman);
\draw [thick, ->] (intro) to node[below] {and Truncation} (carleman);

\draw [thick, -]  (carleman.south)  to node[above] {} (foo1.south);
\draw [thick, -]  (foo1.south)  to node[above] {Temporal} (foo2.south);
\draw [thick, -]  (foo1.south)  to node[below] {Discretization} (foo2.south);
\draw [thick, ->]  (foo2.south)  to node[above] {} (euler.north);

\draw [thick, ->]  (euler)  to node[above] {Linear} (vqls);
\draw [thick, ->]  (euler)  to node[below] {System} (vqls);
\draw [thick, ->]  (vqls)  to node[above] {VQLS} (end);
\draw [thick, ->]  (vqls)  to node[below] {Solution} (end);
\end{tikzpicture}
}
\caption{Illustration of steps involved in transforming a quadratic polynomial ODE~(\ref{eq:qdcs}) into a system of linear algebraic equations~(\ref{eq:linq}) amenable for VQLS. This involves CL and truncation resulting in a system of linear ODEs, time discretization of linear ODEs to generate a system of linear difference equations, and expressing difference equation into a single system of linear algebraic equations with normalized solution and right hand side. While an explicit time discretization is shown in the figure, one can use implicit Euler discretization or other schemes as well.}\label{fig:workflow}
\end{figure}

\paragraph{Forward Euler: }
\paragraph{Step 2:} Applying the forward Euler method to the system (\ref{eq:carlfq}) with step size $h=T/\tsteps$, $M\in\Nr$, results in a set of difference equations
\begin{equation}\label{eq:n_foreuler}
\wvh^{k+1}-(\Id+\Ar_N(kh,\pv)h)\wvh^k=\bv(kh,\pv)h,
\end{equation}
where $\wvh^k$ approximates $\wvh(kh)$ as defined in~\Cref{eq:what} for $k=\{0,\cdots,\tsteps-1\}$ with $\wvh^0=\wvh(0)$. The error introduced by Euler discretization is characterized by the Lemma \ref{lem:euler}.

\paragraph{Step 3:}  The iterative system (\ref{eq:n_foreuler}) can be expressed as a single system of linear equations
\begin{equation}\label{eq:n_forlin}
\overbrace{\left(
  \begin{array}{cccc}
    \Id & 0 & 0 & \cdots \\
    -[\Id+\Ar_N(0,\pv) h] & \Id & 0 & \cdots \\
    0 & 0 & \ddots & \ddots \\
    0 & 0 &  -[\Id+\Ar_N((\tsteps-1)h,\pv) h] & \Id \\
  \end{array}
\right)}^{\Amt(\pv)}\overbrace{\left(\begin{array}{c}
         \wvh^0 \\
         \wvh^1 \\
         \vdots \\
         \wvh^\tsteps \\
       \end{array}\right)}^{\tilde{\wv}}
=\overbrace{\left(\begin{array}{c}
         \wvh(0) \\
         h\bv(0,\pv) \\
         \vdots \\
         h\bv((\tsteps-1)h,\pv) \\
       \end{array}\right)}^{\bvt(\pv)}.
\end{equation}
\paragraph{Backward Euler: }
\paragraph{Step 2:} Similarly, applying the backward Euler method to the system (\ref{eq:carlfq}) with step size $h=T/\tsteps$ yields
\begin{equation}\label{eq:n_backeuler}
\left(\Id - \Ar_N(kh,\pv) h\right)\wvh^{k+1}-\wvh^k=\bv(kh,\pv)h,
\end{equation}
where, $\wvh^k$ approximates $\wvh(kh)$ for $k=\{0,\cdots,\tsteps-1\}$. 

\paragraph{Step 3:}  The iterative system (\ref{eq:n_backeuler}) can similarly be expressed as a single system of linear equations
\begin{equation}\label{eq:n_backlin}
\overbrace{\left(
  \begin{array}{cccc}
    \Id & 0 & 0 & \cdots \\
    -\Id & [\Id-\Ar_N(0,\pv) h] & 0 & \cdots \\
    0 & 0 & \ddots & \ddots \\
    0 & 0 &  -\Id & [\Id-\Ar_N((\tsteps-1)h,\pv) h] \\
  \end{array}
\right)}^{\tilde{\Ar}(\pv)}\overbrace{\left(\begin{array}{c}
         \wvh^0 \\
         \wvh^1 \\
         \vdots \\
         \wvh^\tsteps \\
       \end{array}\right)}^{\tilde{\wv}}
=\overbrace{\left(\begin{array}{c}
         \wvh_{in} \\
         h\bv(0,\pv) \\
         \vdots \\
         h\bv((\tsteps-1)h,\pv) \\
       \end{array}\right)}^{\tilde{\bv}(\pv)}.
\end{equation}

In the VQLS framework the linear system obtained using forward Euler (\Cref{eq:n_forlin}) or backward Euler (\Cref{eq:n_backlin}) is further transformed into the form
\begin{equation}
\tilde{\Am}(\pv)|\tilde{\wv}\ra=|\tilde{\bv}(\pv)\ra,\label{eq:linq}
\end{equation}
where, $|\tilde{\bv}\ra=\tilde{\bv}/\|\tilde{\bv}\|$ and $|\tilde{\wv}\ra=\tilde{\wv}/\|\tilde{\wv}\|$ are the normalized vectors.
The VLQS algorithm then optimizes the parameter $\thetav$ of the ansatz $\Vm(\thetav)$ such that
\begin{equation*}
\tilde{\Am}(\pv)\Vm(\thetav)|\Zev\ra=|\tilde{\bv}(\pv)\ra.
\end{equation*}
The optimal parameter $\thetav^*$ then prepares a solution $|\wvo\ra=\Vm(\thetav^*)|\Zev\ra$ which is an approximation to $|\tilde{\wv}\ra$. The approximation error is characterized by the Lemma \ref{lem:vqlserr}. 

Different cost functions have been proposed for the VQLS algorithm~\cite{VQLS}. This includes global cost function $C_{ug}(\thetav)$ and its normalized version $C_{g}(\thetav)$
\begin{equation}\label{eq:costglob}
C_{ug}(\thetav)=Tr(|\phiv\ra\la\phiv|(\In-|\bvt\ra\la\bvt|))=\la \psiv|\mathbf{H}_{g}|\psiv\ra,\quad C_{g}(\thetav)=\frac{C_{ug}}{\la\phiv|\phiv\ra}=1-\frac{|\la\bvt|\phiv\ra|^2}{\la\phiv|\phiv\ra}
\end{equation}
where $|\phiv\ra=\Amt|\psiv(\thetav)\ra $, $|\psiv(\thetav)\ra=\Vm(\thetav)|\Zev\ra$ and
\begin{equation}\label{eq:Hg}
\mathbf{H}_{g}=\Amt^*\Um_b(\In-|\Zev\ra\la\Zev|)(\Um_b)^*\Amt,
\end{equation}
and similarly unnormalized and normalized local cost functions
\begin{equation}\label{eq:costloc}
C_{ul}(\thetav)=\la \psiv|\mathbf{H}_{l}|\psiv\ra,\quad C_{l}(\thetav)=\frac{C_{ul}}{\la\phiv|\phiv\ra},
\end{equation}
respectively, where the local Hamiltonian $\mathbf{H}_l$ is defined as
\begin{equation}
\mathbf{H}_l=\Amt^*\Um_b\left(\In-\frac{1}{n}\sum_{i=1}^{n}|0_k\ra\la 0_k|\otimes \In_{\tilde{k}}\right)(\Um_b)^*\Amt, \label{eq:Hl}
\end{equation}
with $|0_k\ra$ being a zero state on qubit $k$,  $\In_{\tilde{k}}$ being identity on all qubits except the qubit $k$, and the unitary $\Um_b$ prepares $|\bvt\ra=\Um_b|\Zev\ra$. Computation of these cost functions can be accomplished via Hadamard test circuit and its variations (see~\cite{VQLS} for the details) and rely on linear combination of unitaries (LCU) decomposition of the matrix $\Amt$ as 
\begin{equation}\label{eq:LCU}
\Amt=\sum_{i=1}^{n_l}\alpha_i \Amt_i,
\end{equation}
where $\Amt_i$ are unitary operators and $\alpha_i$ are complex scalar coefficients. One popular approach for unitary operators is based on the Pauli basis formed from the identity $\In$ and the Pauli gates $\Xm$, $\Ym$ and $\Zm$. A matrix $\Amt$ with the size $2^n \times 2^n$ can then be written as a linear combination of elements selected from the basis set
\begin{equation*}
\mathcal{P}_n=\{\Pm_1\otimes \Pm_2 \dots \otimes \Pm_n:\Pm_i\in\{\In,\Xm,\Ym,\Zm\}\}.
\end{equation*}
In this basis, each $\Amt_i\in \mathcal{P}_n$ and its corresponding coefficient $\alpha_i$ can be determined via different numerical approaches~\cite{lcu}.  More recently, alternative tensor product decomposition has been proposed which uses different basis elements to better exploit the underlying structure and sparsity of matrices and results in more efficient decomposition~\cite{liu2021variational,gnanasekaran2024efficient}. We have used this alternate decomposition in the numerical study in~\Cref{sec:num}.

\subsection{Nonlinear BVQPCO}\label{sec:nBVQPCO}
To solve the optimization problem (\ref{eq:gen_optrefor})-(\ref{eq:pc}) in a variational quantum framework, we extend the BVQPCO framework to nonlinear setting, where the outer optimization level iteratively selects $\pv^k$ using a classical black box optimizer based on the cost function evaluated using solution of CL truncated system obtained via VQLS. To do so, let $\Km_0 \in \Rr^{\nx \times \nt_c(\tsteps+1)}$ and $\Km_f\in \Rr^{\nx(\tsteps+1) \times \nt_c(\tsteps+1)}$ be matrices such that
\begin{equation}
\uv_0=\wvh_1^0=\Km_0 \tilde{\wv},\label{eq:}
\end{equation}
and
\begin{equation}
\left(
  \begin{array}{c}
    \wvh_1^0  \\
    \wvh_1^1\\
     \vdots\\
     \wvh_1^M \\
  \end{array}
\right)
=\Km_f \tilde{\wv},
\end{equation}
where, $\wvh^k_1$ is the first $\nx$ components extracted from the vectors $\wvh^k$ for $k=0,\cdots,\tsteps$.  Thus, the two terms in  the cost function \Cref{eq:cstform} can be approximated via the Riemann sum as
\begin{eqnarray}
\la\uv,\Hm\uv\ra_T   &\approx & \sum_{k=0}^{\tsteps} h\la\uv(kh),\Hm\uv(kh)\ra\approx h\la \Km_f\tilde{\wv}, (\Id_{M+1}\otimes\Hm) \Km_f \tilde{\wv}\ra,\\
\la\uv,\hv\ra_T  & \approx &  h \la \mathbf{1}_{\tsteps+1}\otimes\hv,\Km_f \tilde{\wv} \ra,
\end{eqnarray}
where $\mathbf{1}_{\tsteps+1}\in \Rr^{\tsteps+1}$ is a vector of ones. Furthermore
\begin{eqnarray}
\la \Km_f\tilde{\wv}, (\Id_{M+1}\otimes\Hm) \Km_f \tilde{\wv}\ra &=& \|\uv_0\|^2\frac{\la \Km_f\tilde{\wv}, (\Id_{M+1}\otimes\Hm) \Km_f \tilde{\wv}\ra}{\la \Km_0 \tilde{\wv}, \Km_0 \tilde{\wv}\ra},\notag\\
&=& \|\uv_0\|^2\frac{\la  \tilde{\wv} |\Km_f^\tra(\Id_{M+1}\otimes\Hm) \Km_f |\tilde{\wv}\ra}{\la  \tilde{\wv} |\Km^\tra_0 \Km_0| \tilde{\wv}\ra},
\end{eqnarray}
and similarly 
\begin{equation}
\la \mathbf{1}_{\tsteps+1}\otimes\hv,\Km_f \tilde{\wv} \ra = \|\uv_0\| \frac{\la \mathbf{1}_{\tsteps+1}\otimes\hv,\Km_f \tilde{\wv} \ra}{\sqrt{\la \Km_0 \tilde{\wv}, \Km_0 \tilde{\wv}\ra}}=\|\uv_0\| \frac{\la (\mathbf{1}_{\tsteps+1}\otimes\hv)^\tra\Km_f |\tilde{\wv} \ra}{\sqrt{\la  \tilde{\wv} |\Km^\tra_0 \Km_0| \tilde{\wv}\ra}}.
\end{equation}
Thus, the optimization problem (\ref{eq:gen_optrefor})-(\ref{eq:pc})  can be expressed in a variational quantum form as
\begin{eqnarray}
\min_{\pv,\thetav} && f\left(w_1h\|\uv_0\|^2\frac{\la  \Zev|\Vm^{\dagger}(\thetav) \Phim_f\Vm(\thetav)|\Zev\ra}{\la  \Zev |\Vm^{\dagger}(\thetav) \Phim_0\Vm(\thetav) | \Zev\ra}+w_2h\|\uv_0\| \frac{\la  \phiv|\Vm(\thetav)|\Zev\ra}{\sqrt{\la  \Zev |\Vm^{\dagger}(\thetav) \Phim_0\Vm(\thetav) | \Zev\ra}},\pv\right),\label{eq:transopt1}\\
\text{s.t.} && \tilde{\Am}(\pv)\Vm(\thetav)|\Zev\ra=|\tilde{\bv}(\pv)\ra,\label{eq:transopt2}\\
 && g_i(\pv)\leq 0,i=1,\cdots,\nc,\label{eq:transopt3}
\end{eqnarray}
where
\begin{equation}\label{eq:consts}
\Phim_f=\Km_f^\tra(\Id_{M+1}\otimes\Hm) \Km_f, \quad \Phim_0=\Km_0^\tra\Km_0, \quad \phiv=(\mathbf{1}_{\tsteps+1}\otimes\hv)^\tra\Km_f.
\end{equation}

\begin{algorithm}[hbt!]
\begin{algorithmic}[1]
\STATEx Input: Semi-discretized ODE in the form (\ref{eq:qdcs}), cost function (\ref{eq:cstform}), $\{g_i(\pv)\}_{i=1}^{\nc}$, CL truncation level $\nt$, $\Vm(\thetav)$,$\gamma$,$n_{sh}$ and $\epsilon$
\STATEx Output: Optimal parameters: $\pv^*$
\STATE Initialize $k=0$ and $\pv^0=\pv_{in}$.
\STATE Apply CL to (\ref{eq:qdcs}) with truncation level $\nt$ and determine $\Amt(\pv)$, $\bvt(\pv)$ in~\Cref{eq:linq} and associated $\phiv$,$\Phim_f,\Phim_0$ as defined in~\Cref{eq:consts}. 
\STATE Determine the unitary $\Um_{\phiv}$ to prepare $|\phiv\ra$, and LCU decomposition for $\Phim_f=\sum_{i=1}^{n_f}\alpha_{fi}\Phim_{fi}$ and $\Phim_0=\sum_{i=1}^{n_0}\alpha_{0i}\Phim_{0i}$. \label{algoline:prep}
\WHILE{stopping criteria not met}
\STATE Compute LCU $\{\alpha_i(\pv^k),\Amt_i(\pv^k)\}_{i=1}^{n_k}$ of $\Amt(\pv^k)$, and determine an unitary $\Um_b^{k}$ such that $|\bvt(\pv^k)\ra=\Um_b^k|\Zev\ra$.\label{algoline:LCU}
\STATE Compute $\thetav_*^k=VQLS(\{\alpha_i(\pv^k),\Amt_i(\pv^k)\}_{i=1}^{n_k},\Um_b^k,\Vm(\thetav),\gamma,n_{sh})$. \label{algoline:VQLS}
\STATE Let $|\psiv(\thetav_*^k)\ra=\Vm(\thetav_*^k)|\Zev\ra$. Compute $\la\phiv|\psiv(\thetav_*^k)\ra$, $\la\psiv(\thetav_*^k)|\Phim_f|\psiv(\thetav_*^k)\ra$ and $\la\psiv(\thetav_*^k)|\Phim_0|\psiv(\thetav_*^k)\ra$ using associated quantum circuits, and evaluate the cost function~(\ref{eq:transopt1}). \label{algoline:swap}
\STATE Use classical black box optimizer to select next $\pv^{k+1}$ subject to constraints $\{g_i(\pv)\}_{i=1}^{\nc}$.  \label{algoline:opt}
\STATE Check the stopping criterion, e.g., whether $S(\pv^{k},\pv^{k+1})\leq \epsilon$. \label{algoline:conv}
\STATE $k\leftarrow k+1$
\ENDWHILE
\STATE Return $\pv^k$
\end{algorithmic}
\caption{nBVQPCO Algorithm.} \label{algo:nbilevel}
\end{algorithm}

\begin{figure}[!htbp]
\begin{center}
\includegraphics[scale=0.87]{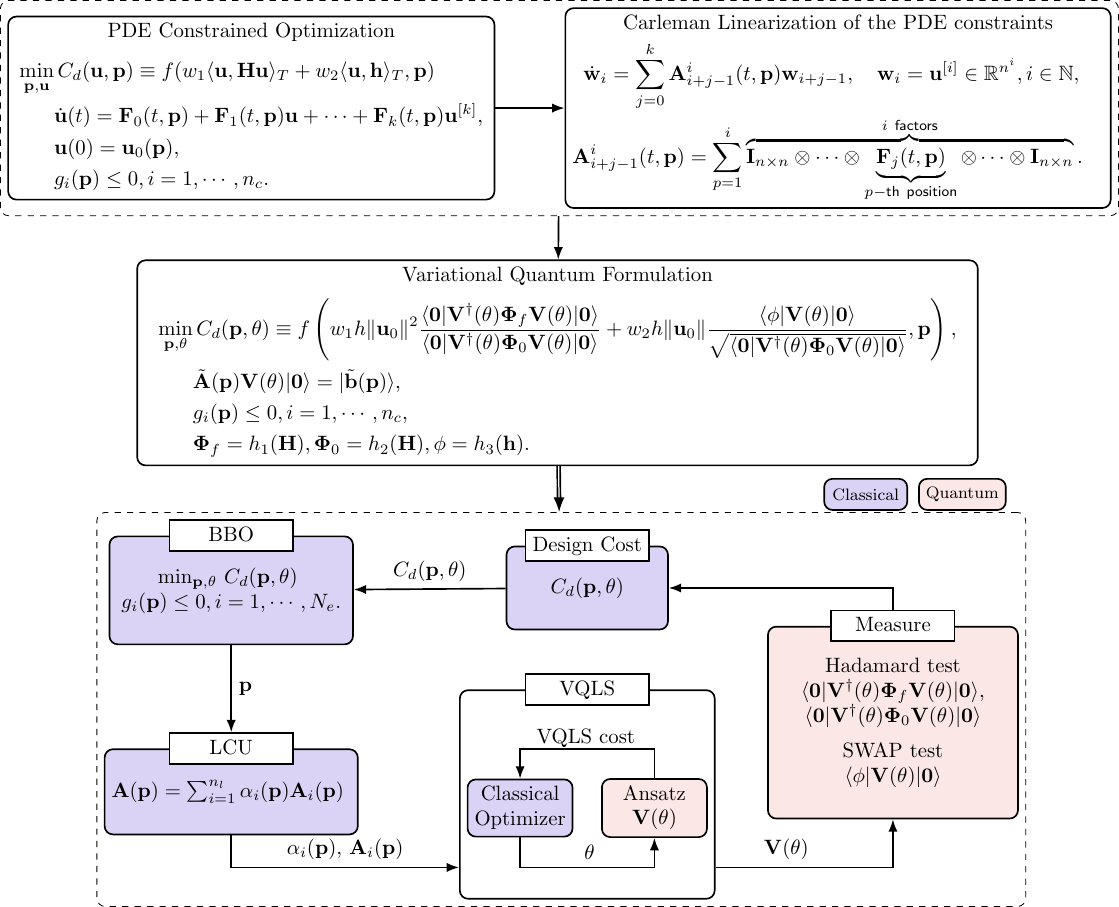}
\caption{Top: Schematic showing transformation of a nonlinear PDE constrained optimization problem into a variational quantum form. Bottom: Flow diagram of the nBVQPCO framework. VQLS uses an inner level optimization to solve the linear system constraints, arising from the discretization of the underlying PDEs, for given design parameters, and evaluate the quantities related to design cost/objective function. A black box optimizer is used for the outer level optimization to select next set of parameters values based on the evaluated design cost.}
\label{fig:nBVQPCO}
\end{center}
\end{figure}

The pseudo code for nBVQPCO is summarized in the~\Cref{algo:nbilevel} and ~\Cref{fig:nBVQPCO} shows the overall flow diagram. Several remarks follow:
\begin{itemize}
    \item Note that VQLS provides solution of the given linear system only upto a normalization constant. By reformulating the cost function expression as discussed above one can make the cost function evaluation independent of the normalization constant.

   \item  The inputs for VQLS in the line (\ref{algoline:VQLS}) includes the linear combination of unitaries (LCU) decomposition of $\Amt(\pv)$, unitary $\Um_b^k$ which prepares $|\bvt(\pv^k)\ra$,  $\Vm(\thetav)$ the selected ansatz, $\gamma$ the stopping threshold and $n_{sh}$ the number of shots used in VQLS cost function evaluation. See~\Cref{sec:num} for the details.  From the relations (\ref{eq:errbounds}), one can select $\gamma$ such that, for instance, for $C_{g}\leq \gamma$, results in a VLQS approximation error $\epsilon\leq \kappa \sqrt{\gamma \log N}$, see~\Cref{lem:vqlserr}.

  \item For LCU decomposition a popular approach is to use Pauli basis as discussed above. We propose to employ a more efficient tensor product decomposition~\cite{gnanasekaran2024efficient}, which exploits the underlying structure and sparsity of matrices such as arising from PDE discretizations, see ~\Cref{sec:num}.

  \item Depending on how $\Amt$ depends on the parameters $\pv$, a parameter dependent LCU decomposition $\{\alpha_i(\pv),\Am_i\}$ can be pre computed once, thus saving computational effort, see~\Cref{sec:invprob} for an example. 
 
   \item For computing $\la\phiv|\psiv(\thetav_*^k)\ra$ in line \ref{algoline:swap} one can use the unitary $\Um_{\phiv}$ as computed in line (\ref{algoline:prep}) and quantum circuit for the SWAP test. Similarly, given the LCU decompositions of $\Phim_f$ and $\Phim_0$ as computed in line (\ref{algoline:prep}) one can express 
       \begin{eqnarray*}
         \la\psiv(\thetav_*^k)|\Phim_f|\psiv(\thetav_*^k)\ra &=& \sum_{i=1}^{n_f}\alpha_{fi}\la\psiv(\thetav_*^k)|\Phim_{fi}|\psiv(\thetav_*^k)\ra,\\
         \la\psiv(\thetav_*^k)|\Phim_0|\psiv(\thetav_*^k)\ra &=&  \sum_{i=1}^{n_0}\alpha_{fi}\la\psiv(\thetav_*^k)|\Phim_{0i}|\psiv(\thetav_*^k)\ra,
       \end{eqnarray*}
       and compute each term in the above sums using quantum circuit associated with the Hadamard test. 
   
   \item For the outer optimization one can utilize any global black box optimization method \cite{surana2024variational}. For instance, one can sample the optimization variables $\pv$ to generate a set of predetermined grid points, evaluate the cost function at those points and select the grid point with the minimum cost. Alternatively, one can use more adaptive techniques such as Bayesian optimization (BO) which is a sequential design strategy for global optimization of black-box functions and uses exploration/exploitation trade off to find optimal solution with minimum number of function calls.

  \item For convergence, line \ref{algoline:conv} uses a step size tolerance, i.e., $S(\pv^{k},\pv^{k+1})=\|\pv^{k+1}-\pv^{k}\|$. However, other conditions can be employed, such as functional tolerance, i.e. the algorithm is terminated when the change in design cost $|C_d^{k+1}-C_d^k|$ is within a specified tolerance $\epsilon$.
   
  \item In the \Cref{sec:erranalysis} we provide a detailed  computational error analysis for CL+VQLS based solution of polynomial ODEs. Since, there are no theoretical results available for query complexity of VQLS, we employ empirically known results along with our rigorous CL+VQLS error analysis to assess potential advantage of the nBVQPCO framework over classical methods.
 

\end{itemize}

\section{Computational Complexity and Error Analysis}\label{sec:erranalysis}
Consider a system of inhomogeneous quadratic polynomial ODEs (\ref{eq:qdcs}) 
\begin{eqnarray}
\dot{\uv}&=&\F_0(t)+\F_1\uv+\F_2\uv^{[2]},\label{eq:qdc}\\
\uv(0)&=&\uv_0\in\Rr^\nx,\notag
\end{eqnarray}
where we have dropped explicit dependence on the parameters $\pv$ and time $t$. Associated with above system, let $\wvc,\wvh$ and $\tilde{\wv}$ be as defined in the Eqns. (\ref{eq:wup}), (\ref{eq:what}) and (\ref{eq:n_forlin}), and let
\begin{equation}\label{eq:wc}
\wv_c=\left(
        \begin{array}{c}
         \wvc(0) \\
         \wvc(h) \\
          \vdots \\
          \wvc(\tsteps h) \\
        \end{array}
      \right).
\end{equation}
See also Fig.~\ref{fig:workflow}, which illustrates the relationship between these different variables. We make following assumptions. 

\begin{assumption}\label{assum1}
For the system (\ref{eq:qdc})
\begin{itemize}
  \item $\F_1,\F_2$ be time independent matrices.
  \item Spectral norms $\|\F_1\|, \|\F_2\|$ and $\|\F_0\|=\max_{t\in[0,T]}\|\F_0(t)\|$ are known, and $\max_{t\in[0,T]}\|\dot{\F}_0\|<\infty$
  \item $\F_1$ be diagonalizable with eigenvalues $\lam_j$ of $\F_1$ satisfying $\textsf{Re}(\lam_n)\leq\cdots\leq \textsf{Re}(\lam_1)<0$.
  \item Let
    \begin{equation}\label{eq:R2}
    R_2=\frac{1}{\lvert \textsf{Re}(\lam_1)\rvert }(\|\uv_0\| \| \F_2\|+\frac{\|\F_0\|}{\|\uv_0\|})<1.
    \end{equation}
\end{itemize}
\end{assumption}

\begin{remark}\label{rem1}
Under the condition (\ref{eq:R2}), system (\ref{eq:qdc}) can be rescaled $\uv\rightarrow \frac{\overline{\uv}}{\gamma }$ to (see Appendix in \cite{liu2021efficient} for details)
\begin{equation}\label{eq:qdct}
\dot{\overline{\uv}}=\overline{\F}_0+\overline{\F}_1\overline{\uv}+\overline{\F}_2\overline{\uv}^{[2]},
\end{equation}
such that the following relations hold
\begin{eqnarray}
   \|\overline{\F}_2\|+\|\overline{\F}_0\|&<& \lvert \textsf{Re}(\lam_1)\rvert , \label{eq:condt1}\\
   \|\overline{\uv}(0)\| &<& 1,\label{eq:condt2}
\end{eqnarray}
where, $\overline{\F}_0=\gamma \F_0$, $\overline{\F}_1=\F_1$, $\overline{\F}_2=\frac{1}{\gamma}\F_2$  and $\gamma$ satisfies, 
\begin{equation}\label{eq:gamma}
\gamma=\frac{1}{\sqrt{\|\uv(0)\|r_{+}}}, \quad \text{with} \quad r_{\pm}=\frac{-\textsf{Re}(\lam_1)\pm\sqrt{(\textsf{Re}(\lam_1))^2-4\|\F_2\|\|\F_0\|}}{2\|\F_2\|}.
\end{equation}
Note that under this rescaling, the value of $R_2$ in (\ref{eq:R2}) remains unchanged.
\end{remark}

Under the Assumptions \ref{assum1}, Lemmas~\ref{lem:cltruc}-\ref{lem:kappa} hold, see \cite{liu2021efficient} for the details.
\begin{lemma}\label{lem:cltruc}
The error $\etav_j(t)=\wv_j(t)-\wvh_j(t)$ for $j\in\{1,2,\cdots,\nt\}$ is bounded as follows
\begin{equation}\label{eq:carlerr}
\|\etav_j(t)\|\leq \|\etav(t)\|\leq t \nt \|\F_2\| \|\uv_{0}\|^{\nt+1},
\end{equation}
where, $\etav(t)=(\etav_1^\tra, \etav_2^\tra,\cdots,\etav_N^\tra)^\tra$.
\end{lemma}

\begin{lemma}\label{lem:euler}
Suppose that error $ \etav(t)$ as introduced in Lemma (\ref{lem:cltruc}) satisfies
\begin{equation}\label{eq:etabound}
\|\etav(t)\|\leq \frac{\|\uv(T)\|}{4},
\end{equation}
then there exists a sufficiently small $h$ such that
\begin{equation}\label{eq:eulererr}
\|\wvh_j^k-\wvh_j(kh)\|\leq\|\wvh^k-\wvh(kh)\|\leq 3 \nt^{2.5}k h^2\left((\|\F_2\|+\|\F_1\|+\|\F_0\|)^2+\|\dot{\F}_0\|\right),
\end{equation}
for all $j\in\{1,\cdots,\nt\}$ and $k\in\{0,1,\cdots,\tsteps-1\}$. Concretely, the choice of $h$ should satisfy
\begin{equation}\label{eq:h}
h\leq \min \left\{\frac{1}{\nt\|\F_1\|},\frac{2(|Re(\lambda_1)|-\|\F_2\|-\|\F_0\|)}{\nt(|Re(\lambda_1)|^2-(\|\F_1\|+\|\F_0\|)^2+\|\F_1\|^2)}\right\}.
\end{equation}
\end{lemma}

\begin{lemma}\label{lem:kappa}
The stiffness $\kappa$ of $\tilde{\Ar}$ in (\ref{eq:linq}), recall arising from application of steps 1-3 to the ODE~(\ref{eq:qdc}), is bounded as follows
\begin{equation}\label{eq:kappa}
\kappa\leq 3(\tsteps+1).
\end{equation}
\end{lemma}

\begin{lemma}\label{lem:vqlserr}
Consider a linear system 
\begin{equation}
\Amt|\psiv\ra=|\bvt\ra,
\end{equation}
where $\Amt\in R^{N\times N}, N=2^n, n\in\Nr$ and $\psiv,\bv\in R^N$. Then following bounds hold for the different VQLS cost functions as defined in the Eqns.~(\ref{eq:costglob}) and~(\ref{eq:costloc})
\begin{equation}\label{eq:errbounds}
C_{ug}\geq \frac{\epsilon^2}{\kappa^2} \quad,C_{g}\geq \frac{\epsilon^2}{\kappa^2\|\Amt\|} \quad,   C_{ul} \geq \frac{\epsilon^2}{n\kappa^2}, \quad C_{l}\ge\frac{\epsilon^2}{n\kappa^2\|\Amt\|},
\end{equation}
where $\kappa$ is condition number  of $\Amt$, and $\epsilon$ is the desired error tolerance 
\begin{equation*}
\epsilon=\rho(\psi,\psi(\thetav^*)),
\end{equation*}
with $\rho$ being the trace norm, and $|\psi(\thetav_*)\ra$ being the approximate VQLS solution.
\end{lemma}
The proof of above result can be found in~\cite{VQLS}.
\begin{assumption}\label{assumption2}
We assume that the query complexity of VQLS scales as
\begin{equation*}
\mathcal{O}(\log^{8.5} N\kappa\log(1/\epsilon)),
\end{equation*}
where $\kappa$ is the condition number and $N$ is size of the matrix $\Am$, and $\epsilon$ is the desired accuracy of the VQLS solution as described in~\Cref{lem:vqlserr}.
\end{assumption}
Note that above assumption is based on an empirical scaling study in~\cite{VQLS} and no theoretical guarantees are available. Furthermore this empirical study used random matrices for which the number of LCU terms scaled as $\mathcal{O}((\log N)^2)$.

\begin{lemma}\label{eq:relnorms}
The trace norm and $l_2$ norm between $|\psiv\ra$ and $|\phiv\ra$ are related as follows
\begin{equation}\label{eq:relnorms1}
\left(1-\frac{\||\psiv\ra-|\phiv\ra\|_2^2}{2}\right)^2+\rho^2(\psiv,\phiv)=1.
\end{equation}
\end{lemma}
See Lemma 2 in \cite{surana2024variational} for the proof.

\begin{restatable}{theorem}{carl}\label{thm:carl}
\label{thm:carl}
Consider the ODE system (\ref{eq:qdc}). Then under the Assumptions \ref{assum1}, for any given $\epsilon>0$, one can choose CL truncation level $\nt$ (see~\Cref{eq:ntref}) and Euler discretrization step size $h$ (see~\Cref{eq:href}) such that
\begin{equation}\label{eq:errEuler}
\||\wv_c\ra-|\tilde{\wv}\ra\|\leq \epsilon,
\end{equation}
where  $|\tilde{\wv}\ra$ is the solution of the linear system (\ref{eq:linq}) and $\wv_c$ is solution of CL as defined in (\ref{eq:wc}).
\end{restatable}
The proof of above theorem is given in the \Cref{sec:proof_thmcarl}. 

\begin{theorem}\label{thm:carlvqls}
Consider the ODE system (\ref{eq:qdc}). Then under the Assumptions \ref{assum1}, for any given $0<\epsilon<1$,  one can choose CL truncation level $\nt$, Euler discretrization step size $h$ and VQLS stopping threshold $\gamma$ such that
\begin{equation}\label{eq:errfull}
\| |\wv_c\ra-|\wvo\ra\|\leq \epsilon,
\end{equation}
where, $|\wvo\ra=\Vm(\thetav^*)|\Zev\ra$ is solution generated by VQLS for system~(\ref{eq:linq}) and $\wv_c$ is solution of CL as defined in (\ref{eq:wc}). Furthermore, under the Assumption \ref{assumption2} the query complexity scales as 
\begin{equation}\label{eq:qcomp}
 \mathcal{O}\left(\log^{8.5}\left(\frac{\nx T^4}{\|\wv_c\|^2\epsilon^2}\right)\frac{T^4}{\|\wv_c\|^2\epsilon^2}\log(1/\epsilon)\right).
\end{equation}
\end{theorem}

\begin{proof}
From Theorem \ref{thm:carl}, choose $\nt,h$ such that
\begin{equation}\label{eq:err1pl}
\||\wv_c\ra-|\tilde{\wv}\ra\|\leq \epsilon/2.
\end{equation}
Let $\gamma$ be
\begin{equation}\label{eq:gamma}
\gamma=\frac{1-(1-\frac{\epsilon^2}{8})^2}{9(\tsteps+1)^2 \ln(\nt_c (\tsteps+1))}<1,
\end{equation}
where, note $\tsteps=T/h$ and $\nt_c$ is the dimension of the truncated CL system (\ref{eq:carlfq}). Let the VQLS algorithm be terminated under the condition
\begin{equation*}
C_{l}\leq \gamma,
\end{equation*}
then from the Lemmas \ref{lem:kappa} and \ref{lem:vqlserr} we obtain
\begin{equation*}
\rho^2(\tilde{\wv},\wvo)=(\epsilon^\prime)^2\leq n\kappa^2\gamma \leq 1-\left(1-\frac{\epsilon^2}{8}\right)^2
\end{equation*}
where $n= \ln(\nt_c (\tsteps+1))$. It then follows from the Lemma \ref{eq:relnorms} that
\begin{equation*}
1-\left(1-\frac{\||\tilde{\wv}\ra-|\wvo\ra\|_2^2}{2}\right)^2\leq 1-\left(1-\frac{\epsilon^2}{8}\right)^2,
\end{equation*}
which implies
\begin{equation}\label{eq:err2p}
\||\tilde{\wv}\ra-|\wvo\ra\|_2\leq \epsilon/2.
\end{equation}
Thus, using the relations (\ref{eq:err1pl}) and (\ref{eq:err2p})
\begin{equation}\label{eq:err1}
\| |\wv_c\ra-|\wvo\ra\|\leq \| |\wv_c\ra-|\tilde{\wv}\ra\|+\||\tilde{\wv}\ra-|\wvo\ra\|\leq \epsilon.
\end{equation}

Furthermore by the Assumption \ref{assumption2} the query complexity of VLQS behaves as $O((\log (\nt_c(\tsteps+1))^{8.5}\kappa\log(1/\epsilon^\prime))$ which using~\Cref{lem:kappa} can be simplified as follows 
\begin{eqnarray}
 && \mathcal{O}(\log^{8.5} (\nt_c(\tsteps+1))\kappa\log(1/\epsilon^\prime)) = \mathcal{O}\left(\log^{8.5}(\nt_c\tsteps)3(\tsteps+1)\log(1/\epsilon)\right)\notag\\
 &=& \mathcal{O}\left(\log^{8.5}\left(\nx^\nt \frac{T}{h}\right)\frac{T}{h}\log(1/\epsilon)\right)\notag\\
 &=& \mathcal{O}\left(\log^{8.5}\left(\frac{\nx T^4}{\|\wv_c\|^2\epsilon^2}\right)\frac{T^4}{\|\wv_c\|^2\epsilon^2}\log(1/\epsilon)\right),
\end{eqnarray}
where, we have used the bounds~\Cref{eq:href} and~\Cref{eq:ntref}, and note that
\begin{equation*}
\|\wv_c\|^2=\sum_{k=0}^{\tsteps}\|\wv(kh)\|^2=\sum_{k=0}^{M}\sum_{i=1}^{\nt}\|\uv(kh)\|^{2i}.
\end{equation*}
%

\end{proof}

\begin{remark}
\Cref{thm:carl} and \Cref{thm:carlvqls} can also be applied to higher order polynomial systems (\ref{eq:gen}), by replacing Assumptions~\ref{assum1} with the Assumptions~\ref{assum2}.
\begin{assumption}\label{assum2}
For the system (\ref{eq:gen})
\begin{itemize}
  \item Let $\F_i, i=0,\dots,k$ be time independent matrices with bounded norms.
  \item  Assume $\tF_1$ is diagonalizable and its eigenvalues satisfy $\tilde{\lam}_{n_k}\leq\cdots\leq\tilde{\lam}_1<0$, where 
    \begin{equation}\label{eq:tF1}
     \tF_1=\left(
        \begin{array}{ccccc}
          \Ar^{1}_{1} & \Ar^{1}_{2}& \Ar^{1}_{3} & \cdots & \Ar^{1}_{k-1} \\
          \Ar^{2}_1 & \Ar^{2}_{2} & \Ar^{2}_{3} & \cdots & \Ar^{2}_{k-1} \\
          0 & 0 & \cdots & \cdots & \cdot  \\
          \vdots & \vdots  & \vdots  & \vdots  & \vdots  \\
          0 & 0 & \cdots & \Ar^{k-1}_{k-2} & \Ar^{k-1}_{k-1} \\
        \end{array}
      \right),
  \end{equation}  
  and $\Ar^{i}_{i+j-1}\in\Rr^{\nx^i\times \nx^{i+j-1}}$ is given by
\begin{equation}\label{eq:Adef1}
\Ar^{i}_{i+j-1}=\sum_{p=1}^{i}\overbrace{\Id_{\nx\times \nx}\otimes\cdots\otimes\underbrace{\F_j}_{p-\textsf{th position}}\otimes\cdots\otimes \Id_{\nx\times \nx}}^{i \textsf{ factors}}.
\end{equation}
which is a generalization of~\Cref{eq:Adef}, see \cite{surana2024efficient} for details.
  \item Let
    \begin{equation}\label{eq:Rk}
R_k=\frac{1}{\lvert \textsf{Re}(\tilde{\lam}_1)\rvert }\left((k-1)\sum_{j=2}^k \|\F_j\|\sqrt{\sum_{i=1}^{k-1}\|\xv_0\|^{2i}}+\frac{\|\F_0\|}{\sqrt{\sum_{i=1}^{k-1}\|\xv_0\|^{2i}}}\right)<1.
\end{equation}
\end{itemize}
\end{assumption}
To apply nBVQPCO framework to the higher order polynomial case,  one can either transform (~\ref{eq:gen}) to the form (~\ref{eq:qdcs}) and then follow procedure described in~\Cref{sec:QCquadODE}  or  apply CL in Step 1 directly to(~\ref{eq:gen}) and then follow the subsequent steps. As shown in \cite{surana2024efficient}, latter is equivalent to former but leads to more compact CL representation. 
\end{remark}

\subsection{Comparison with Classical Approach}\label{sec:cmpexp}
Classically to solve the optimization problem (\ref{eq:gen_optrefor})-(\ref{eq:pc}) a variety of methods can be used. Consistent with the nBVQPCO framework we use forward Euler to solve the ODE (\ref{eq:qdc}) for a given parameter $\pv$, compute the cost function and warp a black box optimizer, e.g. BO.

Applying forward Euler discretization scheme with step size $h=T/\tsteps$ to (\ref{eq:qdc}) results in a set of nonlinear finite difference equations
\begin{equation}\label{eq:direuler}
\uvh^{k+1}=\uvh^{k}+h\left(\F_0(kh)+\F_1\uvh^k+\F_2(\uvh^k)^{[2]}\right),
\end{equation}
where $k=0,\cdots,\tsteps-1$, $\uvh^k\approx \uv(kh)$ and $\uvh^{0}=\uv_0$. Let 
\begin{equation}\label{eq:ucutl}
\uv_c=\left(\begin{array}{c}
        \uv(0) \\
        \uv(h) \\
        \vdots \\
        \uv(\tsteps h)
      \end{array}\right),\quad \tilde{\uv}=\left(\begin{array}{c}
        \uvh^0 \\
        \uvh^1 \\
        \vdots \\
        \uvh^\tsteps
      \end{array}\right).
\end{equation}
The next theorem shows that one can always choose a step size $h$ so that $\tilde{\uv}$ approximates $\uv_c$ to within desired error.
\begin{restatable}{theorem}{classeuler} \label{thm:classeuler}
Consider the ODE system (\ref{eq:qdc}). Then under the Assumptions \ref{assum1}, for any given  $\epsilon>0$ one can choose $h$ (see~\Cref{eq:hebound}) such that
\begin{equation}
\||\uv_c\ra-|\tilde{\uv}\ra\|\leq \epsilon.
\end{equation}
\end{restatable}
For the proof see~\Cref{sec:proof_thmclasseuler}. 

Let  $s$ be the maximum of sparsity of $\F_1$ and $\F_2$. The computational complexity for solving (\ref{eq:direuler}) thus scales as
\begin{eqnarray*}
\mathcal{O}(\nx s\tsteps)&=&\mathcal{O}\left(s\nx \frac{T}{h}\right)=\mathcal{O}\left(s\nx T\left(h_0+\frac{C^2T}{\|\uv_c\|^2\epsilon^2}\right)\right)=\mathcal{O}\left(\frac{s\nx T^2}{\|\uv_c\|^2\epsilon^2}\right),
\end{eqnarray*}
where, we have used the bound (\ref{eq:hebound}) and
\begin{equation*}
\|\uv_c\|^2=\sum_{k=0}^{\tsteps} \|\uv(kh)\|^2.
\end{equation*}
On the other hand, based on~\Cref{thm:carlvqls}, under certain conditions the query complexity of VQLS based explicit Euler scheme scales polylogarithmically w.r.t $\nx$, see~\Cref{eq:qcomp}. Assuming that number of outer loop iterations are similar for both this classical and the nBVQPCO framework, the query complexity of nBVQPCO will scale polylogarithmically with $\nx$, and thus could provide a significant computational advantage for simulation based design problems. 

\begin{remark}
In the comparison above, note we have not accounted for computational effort for LCU decomposition required within the nBVQPCO framework (see line~\ref{algoline:LCU} in ~\Cref{algo:nbilevel}).  This however may not be a significant overhead. As pointed out earlier, depending on how $\Amt$ depends on the parameters $\pv$, a parameter dependent LCU decomposition $\{\alpha_i(\pv),\Am_i\}$ can be pre computed once, thus saving computational effort: we show that via an example in~\Cref{sec:invprob}. Moreover, it may be possible to derive LCU decomposition analytically and implemented efficiently, e.g. when using sigma basis, see~\ref{sec:lcusigma}.  
\end{remark}

\begin{remark}\label{rem:vqlsscaling}
Finally, we will like to remind the reader, the conclusion that nBVQPCO could be advantageous over equivalent classical methods is based on the strong Assumption~\ref{assumption2} regarding VQLS scaling. Recall this scaling is determined empirically using random matrices whose LCU decomposition admit $\mathcal{O}((\log N)^2)$ terms. In future, it will be worthwhile to characterize VQLS scaling behavior for sparse and structured matrices which typically arise from PDE discretization, and refine the comparison above.   
\end{remark}

\section{Numerical Study}\label{sec:num}

\subsection{Inverse Burgers' problem} \label{sec:invprob}
Consider the 1D Burgers' equation, which models convective flow $u(x,t)$,  on a spatial domain $[0,L]$ with Dirichlet boundary conditions
\begin{eqnarray}
  \partial_t u+u\partial_x u &=& \nu\partial_x^2 u +f(x,t),\label{eq:burg1}\\
  u(x,0) &=&u_0(x) , \quad u(0,t)=0, \quad u(L,t)=0, \label{eq:burg2}
\end{eqnarray}
with $f(x,t)$ being the forcing function.  We consider an inverse problem, where given $y(t)=u(x_p,t), t\in [0,T]$ where  $x_p\in[0,L]$ is some fixed point, find $\nu$ such that
\begin{eqnarray}
&&  \min_{\nu} \frac{\frac{1}{T}\int_0^T(y(t)-u(x_p,t))^2dt}{u^2_0(x_p)},\label{eq:opt}\\
\mbox{subject to:} && \mbox{Eqns. } (\ref{eq:burg1})-(\ref{eq:burg2}) \\
&& \nu_{min}\leq  \nu \leq \nu_{max}. \label{eq:optparmct}
\end{eqnarray}
Discretizing the PDE with $\nxx$ spatial grid points with grid size $\Delta x=L/(\nxx+1)$ using a central difference scheme, and letting $u_i(t)=u(x_i,t)$ with $x_i=i\Delta x, i=1,\cdots, \nxx$, leads to
\begin{equation}\label{eq:disBurg}
\dot{u}_i=-u_i\frac{u_{i+1}-u_{i-1}}{2\Delta x}+\nu\frac{u_{i+1}-2u_i+u_{i-1}}{2(\Delta x)^2}+f(x_i,t),
\end{equation}
for $i=2,\cdots,\nxx-1$. Using Dirichlet boundary conditions, we get at $i=1$, and $i=\nxx$
\begin{equation}\label{eq:disBurg1}
\dot{u}_1=-u_1\frac{u_{2}}{2\Delta x}+\nu\frac{u_{2}-2u_1}{2(\Delta x)^2}+f(x_1,t),
\end{equation}
and
\begin{equation}\label{eq:disBurgn}
\dot{u}_{\nxx}=-u_{\nxx}\frac{-u_{\nxx-1}}{2\Delta x}+\nu\frac{-2u_{\nxx}+u_{\nxx-1}}{2(\Delta x)^2}+f(x_{\nxx},t),
\end{equation}
respectively.  Expressing in vector form, the evolution of $\uv=(u_1,\cdots,u_{\nxx})^\tra$ can be expressed as an ODE
\begin{eqnarray}
\dot{\uv}&=&\F_0(t)+\F_1\uv+\F_2\uv^{[2]}, \label{eq:odeburg}
\end{eqnarray}
with an initial condition $\uv(0)=\uv_0=(u_0(x_1),\cdots,u_0(x_{\nxx}))^\tra$, 
\begin{equation}\label{eq:F0F1def}
\F_0(t)=\left(\begin{array}{c}
    f(x_1,t) \\
    f(x_2,t)\\
    \vdots\\
    f(x_{\nxx},t)\\
    \end{array}
    \right), \quad \F_1=\nu\tilde{\F}_1=\nu \frac{1}{2\Delta x^2}\left(
    \begin{array}{ccccc}
      -2 & 1 & 0 & \cdots & 0 \\
       1 & -2 & 1 & 0 & \cdots \\
      0 & 1 & -2 & 1 & \cdots \\
      \vdots & \vdots & \vdots  & \cdots & \vdots \\
      0 & 0 & \cdots & 1 & -2 \\
    \end{array}
  \right),
\end{equation}
and
\begin{equation}\label{eq:F2}
\F_2(i,j)= \frac{1}{2\Delta x}\begin{cases}
  -1 & \text{ if } i=1,j=2 \\
   1 & \text{ if } i=\nxx,j=\nxx-1 \\
   1 & \text{ if } j=(\nxx+1)(i-1), 1<i<\nxx\\
  -1 & \text{ if } j=(\nxx+1)(i-1)+2, 1<i<\nxx\\
   0 & \text{otherwise}
\end{cases}.
\end{equation}
The system (\ref{eq:odeburg}) is in the form (\ref{eq:qdcs}). Applying CL with truncation level $\nt$, we can transform (\ref{eq:odeburg}) to (\ref{eq:carlfq}) with
\begin{equation}\label{eq:AN}
\Ar_\nt(t)=\left(\begin{array}{ccccc}
                \nu\tilde{\Ar}_1^1 & \Ar^1_2 & 0 & \cdots & 0  \\
                \Ar^2_1 & \nu\tilde{\Ar}_2^2 & \Ar^2_3 & \cdots & 0 \\
                0 & \Ar^3_2 & \nu\tilde{\Ar}_3^3 & \Ar^3_4 & 0  \\
                \vdots & \vdots & \vdots & \vdots & \vdots  \\
                0 & 0 & 0 & \Ar^{N}_{N-1} & \nu\tilde{\Ar}^{N}_N\\
              \end{array}
            \right),
\end{equation}
where
\begin{equation}\label{eq:Adef1}
\tilde{\Ar}^{i}_{i}=\sum_{p=1}^{i}\overbrace{\Id_{\nxx\times \nxx}\otimes\cdots\otimes\underbrace{\tilde{\F}_1}_{p-\textsf{th position}}\otimes\cdots\otimes \Id_{\nxx\times \nxx}}^{i \textsf{ factors}}.
\end{equation}
We further decompose $\Ar_\nt$ as
\begin{equation*}
\Ar_\nt=\Ar_{\nt,1}+\nu\Ar_{\nt,2},
\end{equation*}
with
\begin{equation*}
\Ar_{\nt,1}=\left(\begin{array}{ccccc}
                0 & \Ar^1_2 & 0 & \cdots & 0  \\
                \Ar^2_1 & 0 & \Ar^2_3 & \cdots & 0 \\
                0 & \Ar^3_2 & 0 & \Ar^3_4 & 0  \\
                \vdots & \vdots & \vdots & \vdots & \vdots  \\
                0 & 0 & 0 & \Ar^{N}_{N-1} & 0\\
              \end{array}
            \right),\quad  \Ar_{\nt,2}=\left(\begin{array}{ccccc}
                \tilde{\Ar}_1^1 & 0 & 0 & \cdots & 0  \\
                0 & \tilde{\Ar}_2^2 &0 & \cdots & 0 \\
                0 & 0& \tilde{\Ar}_3^3 & 0 & 0  \\
                \vdots & \vdots & \vdots & \vdots & \vdots  \\
                0 & 0 & 0 & 0 & \tilde{\Ar}^{N}_N\\
              \end{array}
            \right).
\end{equation*}

Using the backward Euler time stepping and taking $M=n_t-1$ time steps with step size $h=T/M$, $\tilde{\Am}$ in (\ref{eq:n_backlin}) can be expressed as,
\begin{eqnarray*}
 \tilde{\Am}(\nu) &=& \left(
  \begin{array}{cccc}
    \Id & 0 & 0 & \cdots \\
    -\Id & \Id-[(\Ar_{\nt,1}(0)+\nu\Ar_{\nt,2}(0)) h] & 0 & \cdots \\
    0 & 0 & \ddots & \ddots \\
    0 & 0 &  -\Id & \Id-[(\Ar_{\nt,1}((\tsteps-1)h))+\nu\Ar_{\nt,2}((\tsteps-1)h)) h] \\
  \end{array}
\right) \\
  &=&  \tilde{\Am}_1+\nu\tilde{\Am}_2,
\end{eqnarray*}
where
\begin{align*}
\tilde{\Am}_1 &= \left(
  \begin{array}{cccc}
    \Id & 0 & 0 & \cdots \\
    -\Id & \Id-\Ar_{\nt,1}(0)h & 0 & \cdots \\
    0 & 0 & \ddots & \ddots \\
    0 & 0 & -\Id & \Id - \Ar_{\nt,1}((\tsteps-1)h)) h \\
  \end{array}
\right), \\
\tilde{\Am}_2&=\left(
  \begin{array}{cccc}
    0 & 0 & 0 & \cdots \\
    0 & -\Ar_{\nt,2}(0) h & 0 & \cdots \\
    0 & 0 & \ddots & \ddots \\
    0 & 0 &   & -\Ar_{\nt,2}((\tsteps-1)h)h \\
  \end{array}
\right).
\end{align*}
Consequently, for VQLS one can apply LCU for $\tilde{\Am}_1$, $\tilde{\Am}_2$ once, and then generate a parameter dependent LCU for $ \tilde{\Am} $ for any parameter $\nu$ for implementation of the nBVQPCO framework.  

Furthermore, the optimization problem (\ref{eq:opt}-\ref{eq:optparmct}) can be approximated as
\begin{eqnarray}
&&  \min_{\nu} \frac{h}{T}\frac{\la \yv-\Hm\tilde{\wv},\yv-\Hm\tilde{\wv}\ra}{|\la \hv , \tilde{\wv}\ra|^2},\\
\mbox{subject to: } &&\tilde{\Am}(\nu)\tilde{\wv}=\tilde{\bv}.\\
&& \nu_{min}\leq  \nu \leq \nu_{max},
\end{eqnarray}
where $u_0(x_p)=\la \hv , \tilde{\wv}\ra$ for some appropriate vector $\hv$, $\yv=(y(0),\cdots,y(kh),\cdots,y(\tsteps h))^\tra$ and $\Hm$ is an appropriate matrix which picks sub vector $(\hat{w}^0_{1p},\cdots,\hat{w}^{\tsteps}_{1p})^\tra$ from $\tilde{\wv}$, where $\hat{w}^k_{1p}$ is $p$-th element taken from first $\nxx$ components $\wvh^k_{1}$ of vector $\wvh^k$ for each $k=0,\cdots,\tsteps$. Note that
\begin{eqnarray*}
  \frac{\la \yv-\Hm\tilde{\wv},\yv-\Hm\tilde{\wv}\ra}{|\la \hv , \tilde{\wv}\ra|^2} &=& \frac{\la \yv, \yv\ra}{u^2_0(x_p)}-\frac{2\la \yv, \Hm\tilde{\wv}\ra}{|u_0(x_p)|\la \hv , \tilde{\wv}\ra}+\frac{\la \Hm\tilde{\wv}, \Hm\tilde{\wv}\ra}{|\la \hv , \tilde{\wv}\ra|^2}, \\
   &=&  \frac{\la \yv, \yv\ra}{u^2_0(x_p)}-\frac{2\la \yv, \Hm\tilde{\wv}\ra}{|u_0(x_p)|\la \hv , \tilde{\wv}\ra}+\frac{\la \Hm\tilde{\wv}, \Hm\tilde{\wv}\ra}{|\la \hv , \tilde{\wv}\ra|^2},\\
   &=& \frac{\la \yv, \yv\ra}{u^2_0(x_p)}-\frac{2\|\Hm^T\yv\|}{|u_0(x_p)|\|\hv\|}\frac{\la \Hm^T\yv|\tilde{\wv}\ra}{\la \hv| \tilde{\wv}\ra}+\frac{1}{\|\hv\|^2}\frac{\la \tilde{\wv}| \Hm^T\Hm|\tilde{\wv}\ra}{|\la \hv| \tilde{\wv}\ra|^2}.
\end{eqnarray*}
Thus, the optimization problem can be expressed in the form of (\ref{eq:transopt1}-\ref{eq:transopt3}) as
\begin{eqnarray}
&&  \min_{\nu} \frac{h}{T} \left(\frac{\la \yv, \yv\ra}{u^2_0(x_p)} -\frac{2\|\Hm^T\yv\|}{|u_0(x_p)|\|\hv\|}\frac{\la \Hm^T\yv|\tilde{\wv}\ra}{\la \hv| \tilde{\wv}\ra}+\frac{1}{\|\hv\|^2}\frac{\la \tilde{\wv}| \Hm^T\Hm|\tilde{\wv}\ra}{|\la \hv| \tilde{\wv}\ra|^2}\right),\label{eq:opttrans}\\
\mbox{subject to: } && \tilde{\Am}(\nu)|\tilde{\wv}\ra=|\tilde{\bv}\ra,\\
&& \nu_{min}\leq  \nu \leq \nu_{max}.
\end{eqnarray}
Note that all the inner products can be computed using appropriate quantum circuits once solution from VQLS is available as described in the~\Cref{sec:nBVQPCO}.  


\subsection{LCU Decomposition with Sigma Basis}\label{sec:lcusigma}
For sparse and structured matrices, such as arising from PDE discretization, LCU type decomposition based on sigma basis can offer significant advantage ~\cite{gnanasekaran2024efficient,liu2021variational} as compared with the Pauli basis discussed in~\Cref{sec:QCquadODE}. For example, for the matrices arising from discretrization of Possion equation ~\cite{liu2021variational} and heat equation ~\cite{gnanasekaran2024efficient} the number of LCU terms under sigma basis only scale logarithmically, i.e. as $\mathcal{O}(n)$ with the matrix size $2^n \times 2^n$, compared to the Pauli basis for which the number of terms can vary from $\mathcal{O}(2^n)$ (diagonal matrices) to $\mathcal{O}(2^{2n})$ (dense matrix). 

The sigma basis set $S$ comprises of following operators
\begin{equation}\label{eq:sigmabasis}
S = \{\Sp, \Sm, \Sp\Sm, \Sp\Sm, \Id\},
\end{equation}
where $\Sp=\ket{0} \bra{1}$, $\Sm=\ket{1}\bra{0}$, $\Sp\Sm = \ket{0}\bra{0}$ and $\Sm\Sp=\ket{1}\bra{1}$. Given a matrix $\tilde{\Am}$, its LCU type tensor product decomposition in terms of sigma basis takes a similar form as (\ref{eq:LCU}) but with $\Amt_i\in \mathcal{S}_n$ where
\begin{equation*}
\mathcal{S}_n=\{\sigma_1\otimes \sigma_2 \dots \otimes \sigma_n:\sigma_i\in S\}.
\end{equation*}
Even though some of the operators in $S$ are non-unitary, using the concept of unitary completion, one can still design efficient quantum circuits for computing the global/local VQLS cost functions, see ~\cite{gnanasekaran2024efficient} for details. 

We use a recursive strategy for tensor decomposition for $\tilde{\Am}_1, \tilde{\Am}_2$ under the sigma basis as follows. For simplicity, assume that the forcing function $f(x,t) \equiv 0$. Since, $\F_1, \F_2$ as defined in~\Cref{eq:F0F1def} and~\Cref{eq:F2} are time independent, $\Am_{N}(t) \coloneqq \Am_N$ is also time independent. Assume $n_x=2^s$, $n_t=2^t$.  Note that $\Am_{N} \in \mathbb{R}^{N_c \times N_c}$ where $N_c = \frac{n_x^{N+1}-{n_x}}{n_x-1}$ is not a power of $2$ and it is padded with zero block of size $N_0 = 2^{sN+1} - N_c $ and thus can be represented with a ${sN+1}$ qubit system. Then, $\tilde{\Am}_1$ corresponds to $sN+1+t$ qubits and can be written as 
\begin{align*}
 \tilde{\Am}_1 \coloneqq \tilde{\Am}_1^{(sN+1+t)} &= \hat{\Am}_1^{(sN+1+t)} + \underbrace{\Sp\Sm \otimes \dots \otimes \Sp\Sm}_{t \text{ times}} \otimes \Am_{N,1}\, h ,
\end{align*}
where
\begin{align*}
 \hat{\Am}_1^{(sN+1+t)} &= \begin{pmatrix}
 \hat{\Am}_1^{(sN+t)} & 0 \\
 \Dm_1^{(sN+t)} & \hat{\Am}_1^{(sN+t)}
 \end{pmatrix}, \\
\Dm_1^{(sN+t)} &= \begin{pmatrix}
0 & \dots & -\Id_{sN+1} \\
\vdots & \ddots & \vdots \\
0 & \dots & 0
\end{pmatrix} = -\underbrace{\Sp \otimes \dots \otimes \Sp}_{t \text{ times}} \otimes \Id_{sN+1}.
\end{align*} 
Then
\begin{align*}
\hat{\Am}_1^{(sN+1+t)} &= \Id_2 \otimes \hat{\Am}_1^{(sN+t)} + \Sm \otimes \Dm_1^{(sN+t)}, \\
\hat{\Am}_1^{(sN+2)} &= \begin{pmatrix*}[c]
\Id_{sN+1}- \Am_{N,1}h & 0\\
-\Id_{sN+1} & \Id_{sN+1} - \Am_{N,1}h
\end{pmatrix*} = \Id_2 \otimes \left(\Id_{sN+1}- \Am_{N,1}h \right)- \Sm \otimes \Id_{sN+1}. 
\end{align*}
Similarly, $\Tilde{\Am}_{2}$ can be written as
\begin{align*}
\Tilde{\Am}_2 &= \Id_{n_t} \otimes \Am_{N,2} - \underbrace{\Sp\Sm \otimes \dots \otimes \Sp\Sm}_{t \text{ times}} \otimes \Am_{N,2}.
\end{align*}
Next, we find the tensor decomposition of $\Am_{N,1}, \Am_{N,2}$. Each diagonal block of $\Am_{N,2}$ is of the form
\begin{align*}
\tilde{\Ar}^{i}_{i}=\sum_{p=1}^{i}\overbrace{\Id_{\nx\times \nx}\otimes\cdots\otimes\underbrace{\tilde{\F}_1}_{p-\textsf{th position}}\otimes\cdots\otimes \Id_{\nx\times \nx}}^{i \textsf{ factors}},
\end{align*}
where
\begin{align*}
\Tilde{\F}_1 &\coloneqq \Tilde{\F}_1^{(s)} = \Id_2 \otimes \Tilde{\F}_1^{(s-1)} + \Sm \otimes \Tilde{\Dm}_1^{(s-1)} + \Sp \otimes (\Tilde{\Dm}_1^{(s-1)})^T, \\
\Tilde{\Dm}_1^{(s-1)} &= \underbrace{\Sp \otimes \dots \otimes \Sp}_{s-1 \text{ times}}, \\
\Tilde{\F}_1^{(1)} &= -2\,\Id_2 + \Sm + \Sp.
\end{align*}
With these relations, $\Tilde{\F}_1^{(s)}$ can be written using $2\log n_x +1$ terms. As a result, $\tilde{\Ar}^{i}_{i}$ can be written using $i\left( 2\log n_x +1\right)$ terms out of which the factor $\Id \otimes \dots \otimes \Id$ will repeat $i-1$ times and can be combined into a single term. Hence, $\tilde{\Ar}^{i}_{i}$ can be written using $2i \log n_x + 1$ terms. Finally $\Am_{N,2}$ can be decomposed using $\sum_i \left(2i \log n_x + 1\right) = N + N (N+1) \log n_x$ terms. 

Next, we focus on the $\Am_{N,1}$ term. Since $f(x,t) = 0$, the sub-diagonal blocks $\Am^{i+1}_i$ are zero. Each of the super-diagonal block $\Am_{i+1}^i \in \mathbb{R}^{n_x^i \times n_x^{i+1}}$. The first term $\Am_{2}^1 = \F_2 \in \mathbb{R}^{n_x\times n_x^2}$ can be split into square blocks $\F_2^i$ of size $n_x \times n_x$
\begin{align*}
\Am_{2}^1 = \F_2  = \begin{pmatrix}
\F_2^1 & \F_2^2 & \dots & \F_2^{n_x}
\end{pmatrix}.
\end{align*}
Let's say that each block $\F_2^i, i=1,\cdots,\nxx$ can be decomposed with at most $C$ terms where $C$ is a small constant. We will show below that $C\leq 2$ and there are no repeat factors in the decomposition of $\F_2^p$ and $\F_2^q$ for any $p\neq q$. Then, representing $\Am_{2}^1$ in $\Am_{N,1}$ requires $C\,n_x$ terms; one for each $\F_2^i$. Next, we can write
\begin{align*}
\Am_3^2 &= \Id \otimes \F_2 + \F_2 \otimes \Id \\
&= \begin{pmatrix}
\F_2 & & \\
& \F_2 & \\
& & \ddots & \\
& & & \F_2
\end{pmatrix} + \begin{pmatrix}
\F_2^1 \otimes \Id  & \F_2^2 \otimes \Id & \dots & \F_2^{n_x} \otimes \Id
\end{pmatrix}.
\end{align*} 
Each block $\F_2^i \otimes \Id $ is of size $n_x^2 \times n_x^2$ and can be written directly using the tensor decomposition of $\F_2^i$. Then, $\F_2 \otimes \Id$ requires $C\,n_x$ terms. On the other hand, each of the diagonal blocks of $\Id \otimes \F_2$ requires $n_x$ terms. Thus, representing $\Am_3^2$ in $\Am_{N,1}$ requires $C\left(n_x + n_x^2\right)$ terms. In general, we can prove by induction that the number of terms needed are as follows: 
\begin{align*}
\Am_{2}^1 &:C\, n_x \text{ terms}, \\
\Am_{3}^2 &: C \left(n_x +n_x^2\right) \text{ terms}, \\
\Am_{4}^3 &:  C \left(n_x + n_x^2 + n_x^3\right) \text{ terms}, \\
\vdots & \qquad \vdots \\
\Am_{N}^{N-1} &: C\sum_{i=1}^{N-1} n_x^i \text{ terms}.
\end{align*}
Thus, $\Am_{N,1}$ has a total of $\sum_{j=1}^{N-1} \sum_{i=1}^{j}C\, n_x^i = \frac{C\,n_x}{(n_x-1)^2}\left( n_x^N-1 - N(n_x-1) \right)$. Putting it all together, 
\begin{align*}
\tilde{\Am}_1 &: \log n_t + 1 + 2\left( \frac{C\,n_x}{(n_x-1)^2}\left( n_x^N-1 - N(n_x-1) \right)\right) \text{ terms}, \\
\tilde{\Am}_2 &: 2 \left( N + N (N+1) \log n_x\right) \text{ terms} .\\
\end{align*} 

Finally, we look at the decomposition of the blocks $\F_2^i$ and determine the constant $C$. 
The first and last block of $\F_2$ have one non-zero entry with following structure
\begin{align*}
\F_2^1 &= \frac{1}{2\Delta x} \begin{pmatrix}
0 & -1 & 0 & \dots & 0 \\
0 & 0 & 0 & \dots & 0\\
 &   \vdots & & \vdots & \\
 0 & 0 & 0 & \dots & 0\\
\end{pmatrix} = \frac{-1}{2\Delta x} \underbrace{\Sp\Sm \otimes \dots \otimes \Sp\Sm}_{s-1 \text{ times}} \otimes \Sp, \\
 \F_2^{n_x} &= \frac{1}{2\Delta x} \begin{pmatrix}
0 &  \dots & 0 & 0 & 0\\
0  & \dots & 0 & 0 & 0\\
 &   \vdots & \vdots &  \\
 0 &  \dots & 0 & 1 & 0\\
\end{pmatrix} = \frac{1}{2\Delta x} \underbrace{\Sm\Sp \otimes \dots \otimes \Sm\Sp}_{s-1 \text{ times}} \otimes \Sm,\\
\end{align*}
and thus have a single decomposition term, i.e. $C=1$. The remaining $\F_2^k$ blocks have exactly two non-zero entries and can be decomposed into two terms (i.e., $C=2$), which follows from the result below.  

\begin{theorem}\label{thm:sigma_decomp}
Let $A_s$ be a $2^s \times 2^s$ matrix with a single non-zero entry in position $(r,c)$ $r,c \in \{0, 1, \dots, n_x-1\}$. Let the binary representation of $r = \sum_{p=0}^{s-1} 2^p q_r(p)$ and $c = \sum_{p=0}^{s-1} 2^p q_c(p)$.  Then, the tensor decomposition of $A$ in sigma basis is given by $A_s = A_s(r,c) \left(\otimes_p \sigma (q_r(p), q_c(p))\right)$ where,
\begin{align*}
\sigma (q_r(p), q_c(p)) &= \begin{cases}
  \Sp\Sm & \text{ if } q_r(p)=0,q_c(p)=0 \\
   \Sp & \text{ if } q_r(p)=0,q_c(p)=1 \\
   \Sm & \text{ if } q_r(p)=1,q_c(p)=0\\
  \Sm\Sp & \text{ if } q_r(p)=1,q_c(p)=1
\end{cases}.
\end{align*} 
\end{theorem}
\begin{proof}
This can be proved by induction. Trivially true for $s=1$, by the definition of sigma basis. Assume, the statement holds for $s-1$. Then, 
\begin{align*}
A_s = A_s(r,c) \left(\begin{bmatrix}
a & b \\ c & d
\end{bmatrix} \otimes A_{s-1} \right)
\end{align*}
where $a, b,c, d \in \{0,1 \}$. Since, there is only one non-zero entry in $A_s$, only one of $a, b, c$ or $d$ can be non-zero. If $ r\,//\,2^{s-1} = c\,//\, 2^{s-1} = 0$, where $//$ corresponds to integer division, then $a=1$. Equivalently, using the binary representation of $r, c$, if $q_r(s-1) = q_c(s-1) = 0$, then $a=1$ and the corresponding basis is $\Sp\Sm$. Similarly, $b=1$ if $q_r(s-1) = 0, q_c(s-1)=1$, $c=1$ if $q_r(s-1)=1, q_c(s-1)=0$, and $d=1$ if $q_r(s-1)=q_c(s-1)=1$. Combining this with the induction hypothesis completes the proof. 
\end{proof}

In general, the tensor decomposition of a square matrix $A_{rc}$ in sigma basis with a single non-zero entry at $(r,c)$ position for $r,c \in \{0, 1, \dots, n_x-1\}$ can be determined as given in~\Cref{thm:sigma_decomp}. 

To demonstrate the efficacy of using sigma basis compared to the Pauli basis, we compute the number of terms needed in the LCU decomposition for different matrix sizes for $\tilde{\Am}_1, \tilde{\Am}_2$. To vary the matrix size we consider different number of spatial/temporal discretization points $\nxx$ and $\tstepst$, respectively and different values of CL truncation level $N$. For Pauli basis we used the matrix slicing method proposed in~\cite{lcu} to determine the number of terms, while for sigma basis we used analytical expressions derived above. Figure~\ref{fig:lcucomp} shows the comparison and highlights significantly better scaling for the sigma basis. For example, for $\nxx=\tstepst=4$, $\nt=4$, LCU with Pauli basis generates $276,000$ terms whereas sigma basis gives rise to only $111$ terms. Furthermore, computation of the Pauli LCU decomposition was too time consuming as the problem size increases as indicated by the missing data points for $\nxx=\tstepst=8,16$ and $\nt>2$.   

\begin{figure}
\centering
\includegraphics[scale=0.5]{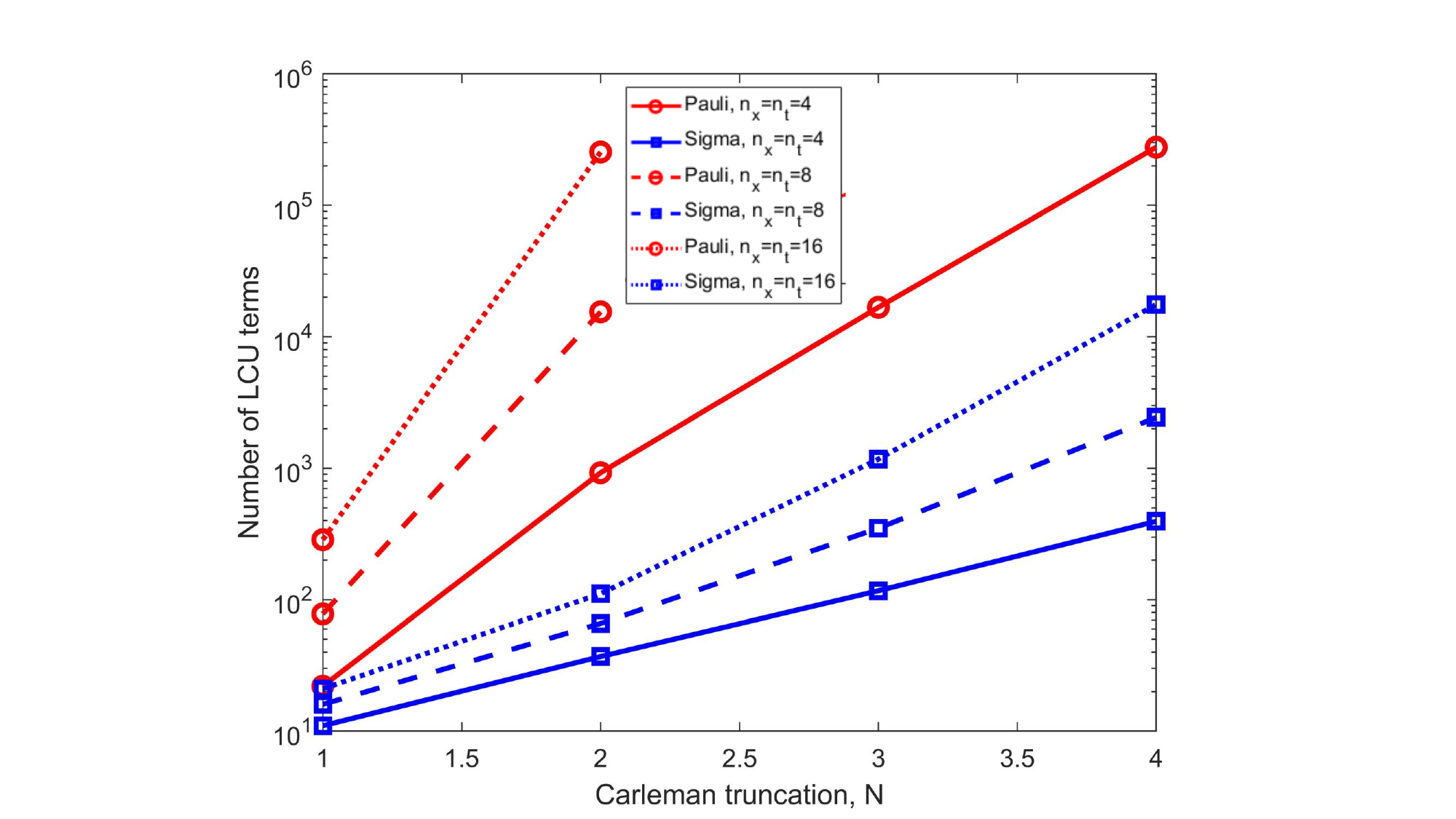}
\caption{Comparison of number of LCU terms using sigma basis (blue curve) and Pauli basis (red curves) for different values of CL truncation levels $N$ and number of spatial/temporal discretization points $\nxx$ and $\tstepst$.}
\label{fig:lcucomp}
\end{figure}

\subsection{Implementation}
In this section, we demonstrate the application of our nBVQPCO framework to solve the inverse problem described in ~\Cref{sec:invprob}.  We take $L=0.5$ and $T=0.35$, and assume that the forcing function $f(x,t) \equiv 0$. The problem is discretized on a spatial grid with $\nxx=4$ grid points leading to $\Delta x =0.1$. For time discretization we take $n_t=8$ time steps resulting in $\Delta t = 0.05$. We consider an initial condition of the form $u_0(x) = c\sin(k(x-\Delta x)), k=\frac{2\pi}{L}$, which upon discertization takes the form
\begin{equation}
 \uv_0 = c \left(\begin{array}{cc}
    \sin(0)   \\
     \sin(k\Delta x) \\
     \vdots\\
     \sin((\nxx-1)k\Delta x)
 \end{array}\right).  \label{eq:u0des}
\end{equation}
For simplicity we choose $c$ such that $\|\uv_0 \|=1$. To simulate the measurement data $y(t),t\in[0,T]$, we take the measurement point to be $x_p=x_2$ where $x_i=i\Delta x$. We integrate the nonlinear ODEs~(\ref{eq:odeburg}), and generate the measurement data $y(t_i)=u_2(t_i),t_i=ih, i=0,\cdots,n_t-1$ using $\nu=0.07$. 
In our numerical studies we investigate $N=1$ and $N=2$ truncation for the CL. 

The nBVQPCO framework was implemented using the PennyLane software framework for quantum computing. PennyLane's \texttt{lightning.qubit} device was used as the simulator, \texttt{autograd} interface was used as the automatic differentiation library and the adjoint method was used for gradient computations. We do not consider any shot, device or measurement noise throughout this study. 

\subsubsection{VQLS Implementation and Results}
In this section we provide details related to VQLS implementation including state preparation circuit,  and the selected ansatz, optimizer, and LCU decomposition approach.

\paragraph{State Preparation:} Here we describe circuit to construct $\tilde{\bv}$ in (\ref{eq:n_backlin}). Given padding of $\Am_N$ as discussed in~\Cref{sec:lcusigma}, we equivalently pad $\bvt$ which takes the form
\begin{equation}\label{eq:btildeburg}
\bvt=\left(\begin{array}{c}
           0\\
         \wvh(0) \\
         0 \\
         h\bv(0,\pv)\\
         \vdots \\
         0 \\
          h\bv((\tstepst-2)h,\pv) \\
       \end{array}\right)=\left(\begin{array}{c}
           0\\
         \wvh(0) \\
         0 \\
        0 \\
         \vdots \\
         0 \\
         0 \\
       \end{array}\right),
\end{equation}
where $\wvh(0) = (\uv_0 \quad (\uv_0^{[2]})^\tra \cdots \quad (\uv_0^{[\nt]})^\tra)^\tra$ with $\uv_0$ defined in~\Cref{eq:u0des}. Since $f(x,t)=0$, implies $\F_0(t)=0$ and thus $\bv(t,\pv)=0$ in~\Cref{eq:what}, and the second equality above follows.
Since, all other entries are zero, it is sufficient to focus on state preparation of the form $
\frac{1}{\|\wvh(0)\|}\left(\begin{array}{cc}
     0 \\
     \wvh(0)\\ 
\end{array}\right)$.
We  generate this state using a sequence of RY, (multi-) controlled RY gates and controlled applications of the circuit for generating $\uv_0$ as follows.

State preparation for $\uv_0$ can be accomplished using RY, CNOT and PauliZ gates as shown in~\Cref{fig:sinkx} for different values of $\nxx$. Next, starting from the state $\ket{0}^{\nx_x}$, we first create the state $\ket{\hat{\wv}^1(0)}= \frac{1}{\sqrt{N}}(0, e_1^\tra, (e_1^{[2]})^\tra,\dots, (e_1^{[\nt]})^\tra)^\tra$ where $e_1 = [1, 0, \dots, 0]^\tra\in \Rr^{\nx_x}$ using RY gates as shows in~\Cref{fig:carl_state_part1}. Finally, by adding the circuit to prepare $\uv_0$ and its controlled application as illustrated in~\Cref{fig:carl_state_part2}, we  sequentially create the states
\begin{align*}
\ket{\hat{\wv}^1(0)} &\rightarrow \frac{1}{\sqrt{N}}(0^\tra, \uv_0^\tra, \uv_0^\tra\otimes e_1^\tra, \dots, \uv_0^\tra \otimes (e_1^{[N-1]})^\tra)^\tra \\
&\rightarrow \frac{1}{\sqrt{N}}(0^\tra, \uv_0^\tra, (\uv_0^{[2]})^\tra, \dots, (\uv_0^{[2]})^\tra \otimes (e_1^{[N-2]})^\tra)^\tra \\
&\rightarrow \frac{1}{\sqrt{N}}(0^\tra, \uv_0^\tra, (\uv_0^{[2]})^\tra, (\uv_0^{[3]})^\tra, \dots, (\uv_0^{[3]})^\tra\otimes (e_1^{[N-3]})^\tra)^\tra \\
\vdots \\
&\rightarrow \frac{1}{\sqrt{N}}(0^\tra, \uv_0^\tra, (\uv_0^{[2]})^\tra, (\uv_0^{[3]})^\tra, \dots, (\uv_0^{[N]})^\tra)^\tra=\frac{1}{\|\wvh(0)\|}\left(\begin{array}{cc}
     0 \\
     \wvh(0)\\ 
\end{array}\right).\\
\end{align*}
Note that since $\|\uv_0\|=1$ and $\|\uv_0^{[i]}\|=1$ for $i=1,2,\cdots,\nt$, the states generated above are normalized as required. Finally, since the forcing term $f(x,t)=0$ in our example, simply adding empty $\log_2{n_t}$ 
wires completes state preparation of the state $\bvt$.

\begin{figure}
\centering
\begin{subfigure}[b]{0.25\textwidth}
        \centering
    \includegraphics[scale=0.4]{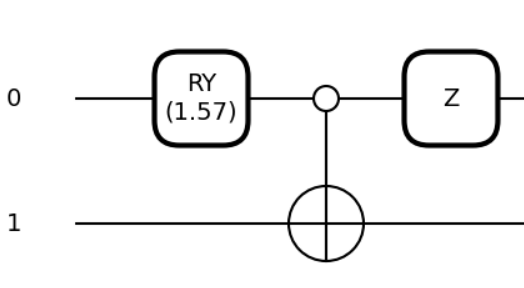}
    \caption{$n_x=4$}
\end{subfigure}
\begin{subfigure}[b]{0.7\textwidth}

    \includegraphics[scale=0.4]{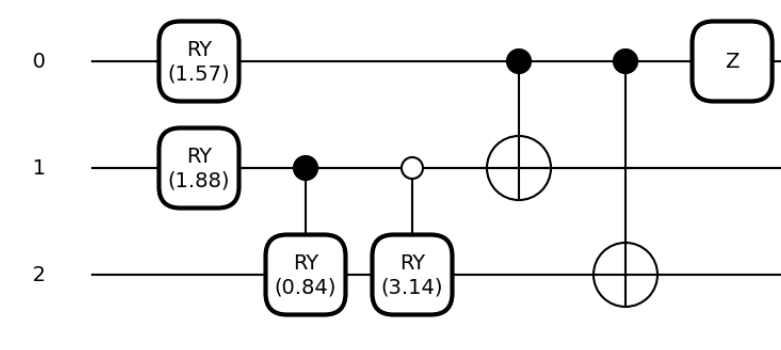}
        \caption{$n_x=8$}
    \end{subfigure}
\caption{Quantum circuit to prepare the initial condition $\uv_0$~\Cref{eq:u0des} for $n_x=4$ and $n_x=8$ grid points. }
\label{fig:sinkx}
\end{figure}

\begin{figure}
\centering
\begin{subfigure}[b]{0.35\textwidth}
        \centering
    \includegraphics[scale=0.4]{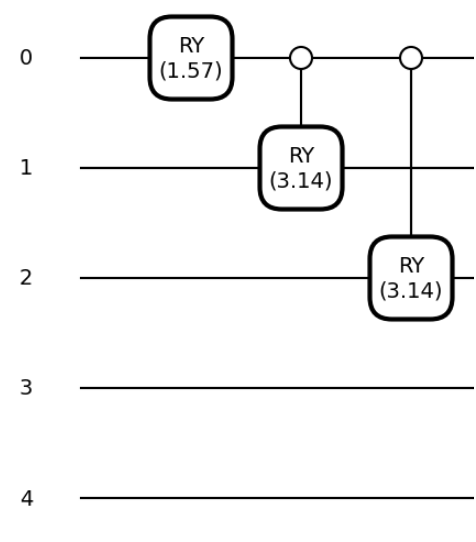}
    \caption{$N=2$}
\end{subfigure}
\begin{subfigure}[b]{0.6\textwidth}

    \includegraphics[scale=0.4]{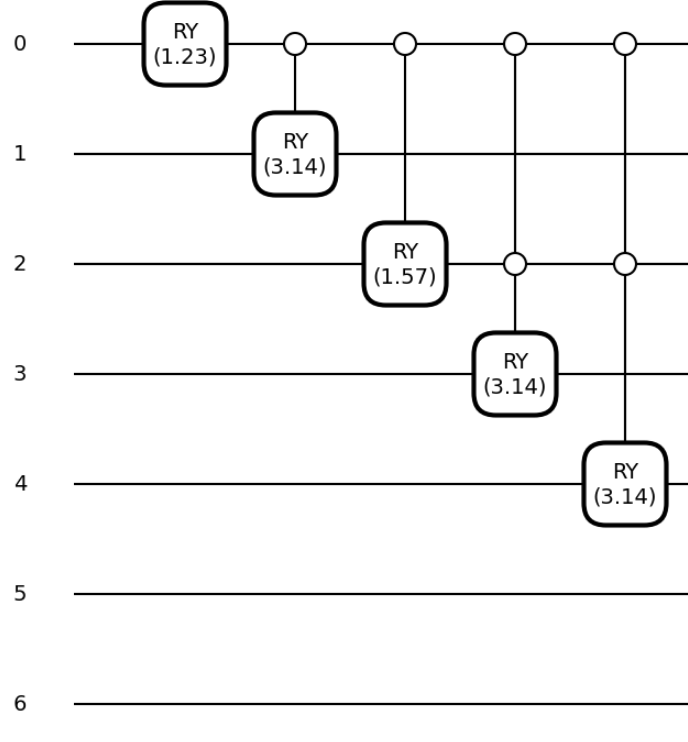}
        \caption{$N=3$}
    \end{subfigure}
\caption{Quantum circuit to prepare the state $\frac{1}{\sqrt{N}}(0^\tra, e_1^\tra, (e_1^{[2]})^\tra,\dots, (e_1^{[\nt]})^\tra)^\tra$ for $n_x=4$ with (a) $N=2$ and (b) $N=3$. }
\label{fig:carl_state_part1}
\end{figure}

\begin{figure}
\centering
\begin{subfigure}[b]{0.32\textwidth}
        \centering
    \includegraphics[scale=0.29]{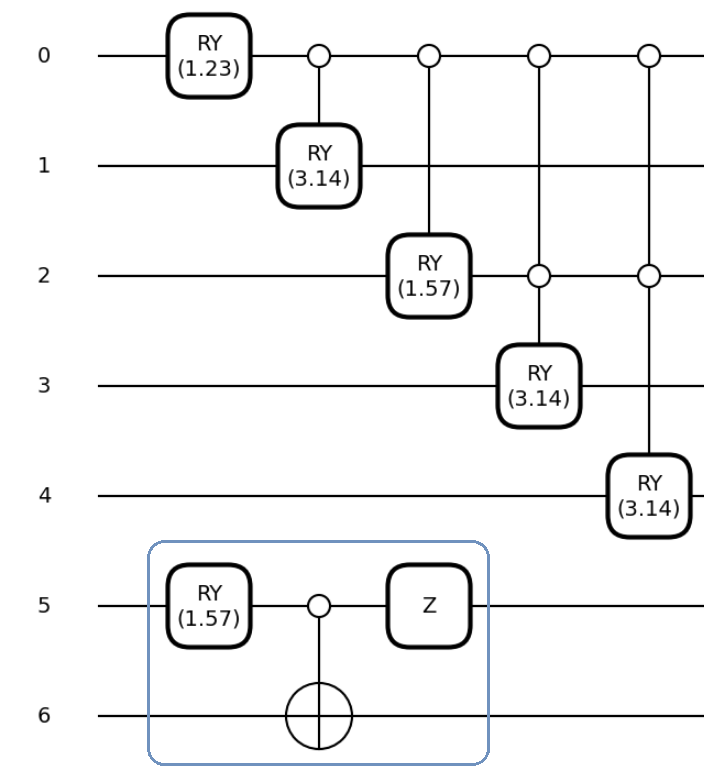}
    \caption{$(0^\tra, \uv_0^\tra, \uv_0^\tra\otimes e_1^\tra, \uv_0^\tra \otimes (e_1^{[2]})^\tra)^\tra$}
\end{subfigure}
\begin{subfigure}[b]{0.65\textwidth}

    \includegraphics[scale=0.32]{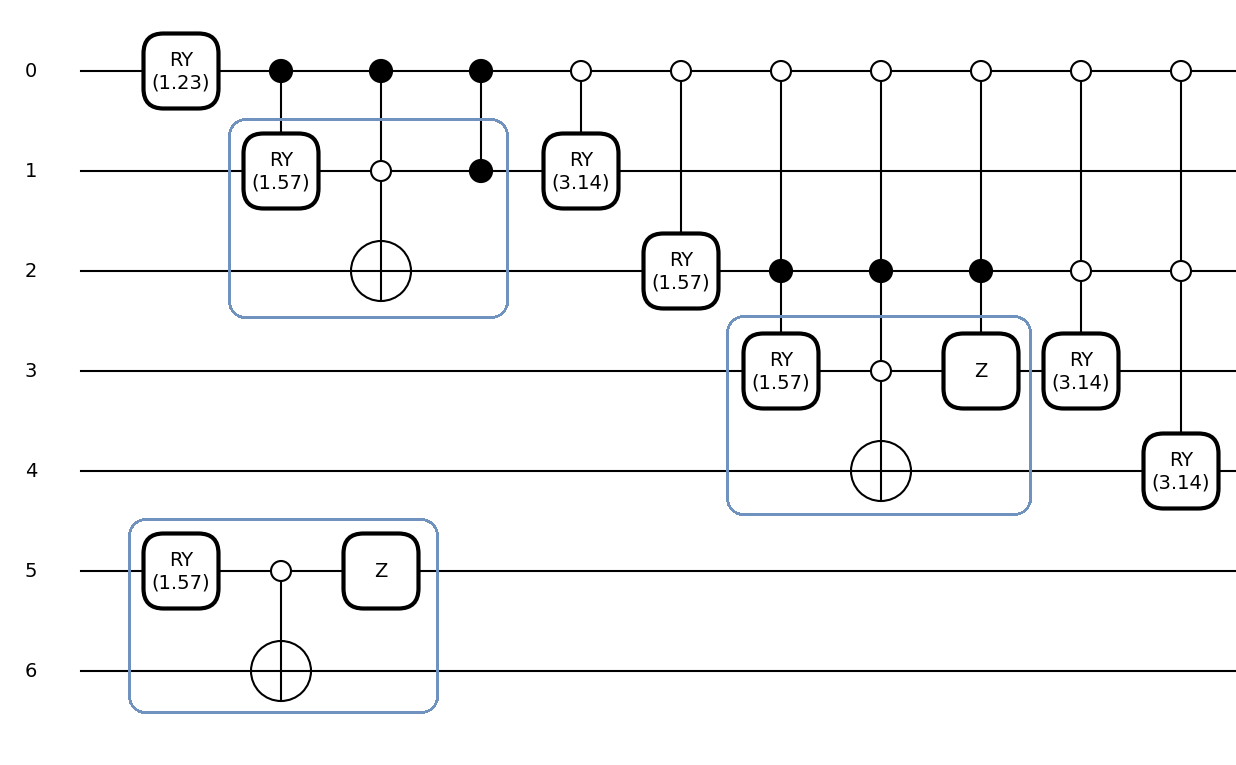}
        \caption{$(0^\tra, \uv_0^\tra, (\uv_0^{[2]})^\tra, (\uv_0^{[2]})^\tra \otimes e_1^\tra)^\tra$}
    \end{subfigure}
    
\begin{subfigure}[b]{1.0\textwidth}
\centering
    \includegraphics[scale=0.4]{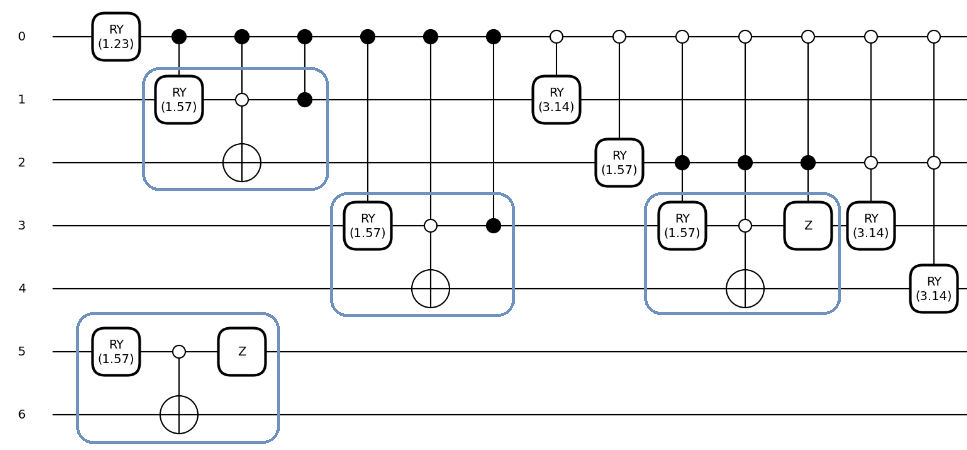}
        \caption{$(0^\tra, \uv_0^\tra, (\uv_0^{[2]})^\tra, (\uv_0^{[3]})^\tra)^\tra$}
    \end{subfigure}
\caption{Sequential construction of the quantum circuit to prepare the desired state 
$(0, \uv_0^\tra, (\uv_0^{[2]})^\tra, (\uv_0^{[3]})^\tra)^\tra$ with $N=3$, $n_x=4$ and $\uv_0$ as defined in~\Cref{eq:u0des}.
}
\label{fig:carl_state_part2}
\end{figure}

\paragraph{Ansatz:} We use a modified version of the ansatz circuit 9 from~\cite{sim2019expressibility}  with 3 repetitions inspired by the work in~\cite{demirdjian2022variational}. This is shown in~\Cref{fig:n_ansatz}.
\begin{figure}
\centering
\includegraphics[scale=0.5]{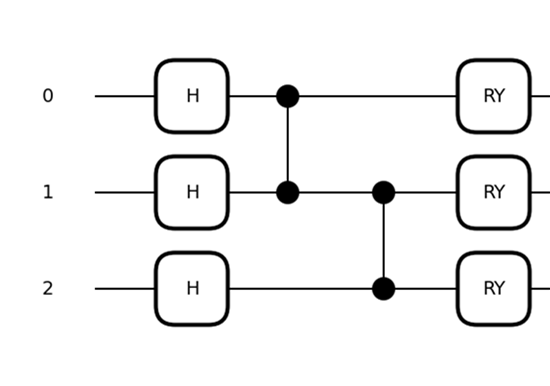}
\caption{Modified real version of circuit 9 from~\cite{sim2019expressibility} that is used as the VQLS ansatz.}
\label{fig:n_ansatz}
\end{figure}

\begin{figure}[hbt!]
   \centering
  \begin{subfigure}[b]{\textwidth}
    \centering  \includegraphics[scale=0.5]{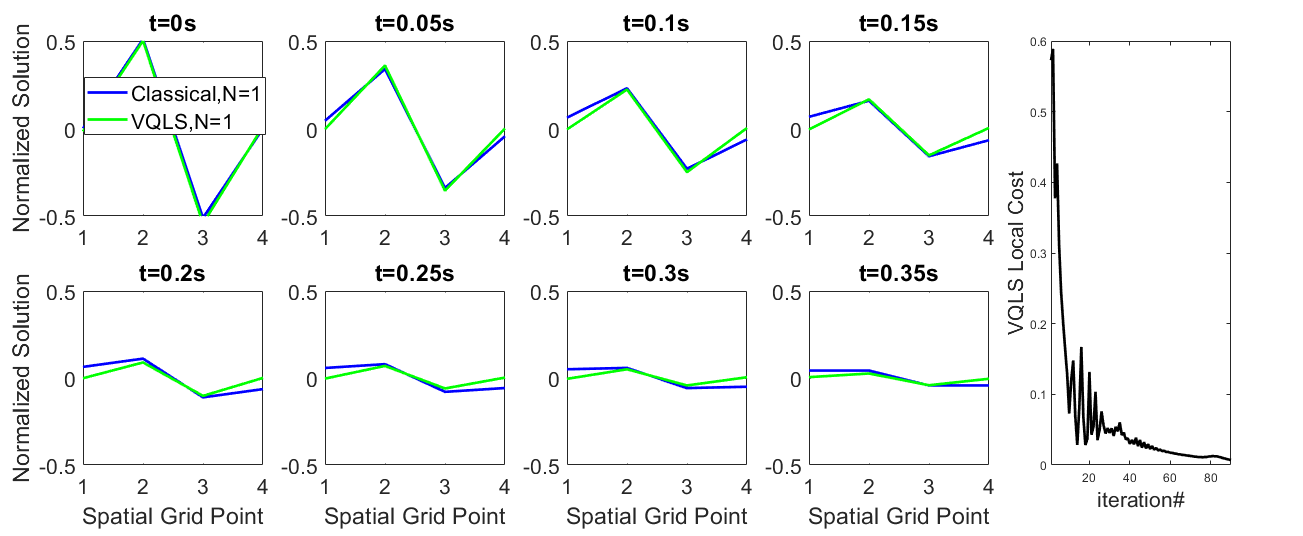}
  \caption{CL truncation level $\nt=1$.}
  \end{subfigure}
  \vfill 
  \begin{subfigure}[b]{\textwidth}
   \centering
    \includegraphics[scale=0.5]{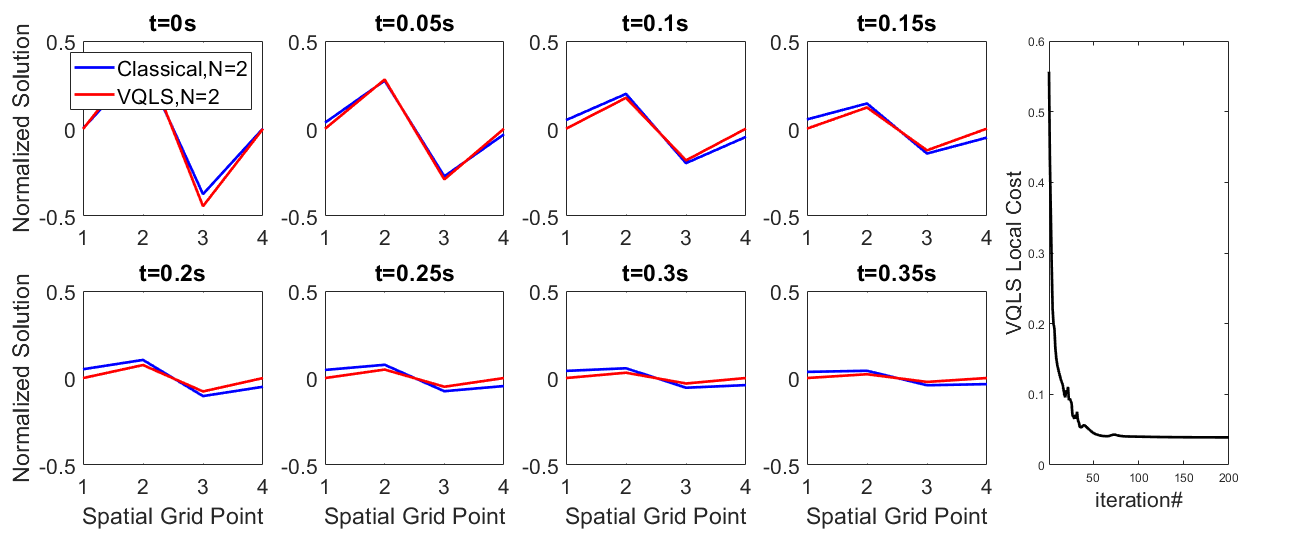}
    \caption{CL truncation level $\nt=2$.}
  \end{subfigure}
  \caption{Comparison of normalized solution, obtained classically and via VQLS, of the implicit linear system (\ref{eq:n_backlin}) for the inverse Burgers problem with $\nu=0.07$ and for two CL truncation levels.}
  \label{fig:vqlsN1N2}
\end{figure}

\paragraph{LCU Decomposition:} As discussed in~\Cref{sec:lcusigma}, we use sigma basis based tensor product decomposition due to its computationally efficiency compared to the Pauli basis for our application. Since, sigma basis is non-unitary, we used circuits based on unitary completion~\cite{gnanasekaran2024efficient} to implement the Hadamard test required for computing the VQLS cost functions.

\paragraph{Optimizer:} There are several choices of optimizers which can be used with VQLS \cite{pellow2021comparison}. We used the Adagrad optimizer from PennyLane with a step size of $0.8$. 
The convergence criteria for VQLS was set to a maximum of $200$ iterations. The optimizer was initialized randomly with samples taken from the Beta distribution with shape parameters $\alpha=\beta=0.5$. 

\paragraph{VQLS Results:} We study the convergence of VQLS for two Carleman truncation levels $N=1,2$ using local VQLS cost function. The VQLS results and convergence of the VQLS cost functions are shown in the~\Cref{fig:vqlsN1N2}. The VQLS solution is compared with the solution obtained by classically solving the underlying linear system~(\ref{eq:n_backlin}). Qualitative comparison of the two solutions indicate that VQLS produces reasonable solutions to the problem.  To characterize VQLS solution quality we define aggregated solution error as time average of norm of normalized solution error at each time step, i.e.
\begin{equation}\label{eq:err}
\mathcal{E}=\frac{1}{n_t}\sum_{i=1}^{n_t} \|\tilde{\wv}^i-\overline{\wv}^i\|,
\end{equation}
where $\tilde{\wv}^i$ is the vector of components taken from $\ket{\tilde{\wv}}$ corresponding to classical solution at the time step $i$, and similarly $\overline{\wv}^i$ is the vector of components taken from $\ket{\overline{\wv}}$ corresponding to VQLS solution at the $i$-th time step. Figure~\ref{fig:compCostN1N2} shows the error $\mathcal{E}$ for different values of $\nu$ and two CL truncation levels. The VQLS solution with $\nt=2$ truncation level has lower error compared to the truncation level $\nt=1$. 

\subsubsection{nBVQPCO Results}\label{sec:conv_design}
Next we study the performance of the nBVQPCO framework. The true optimum of the inverse problem is  $\nu^*=0.07$ as we used that value to generate the simulated measurement data $y(t)$. Since the optimization variable is a scalar, we use exhaustive search as the black box optimizer. In this approach we uniformly sample $\nu$ in the range $[\nu_{min},\nu_{max}]$, use CL+VQLS to generate the normalized solution and compute the design cost~(\ref{eq:opttrans}).  As the curves in Fig.~\ref{fig:compCostN1N2}b indicate, our algorithm returns the minimum solution at $\nu=0.06$ for both CL truncation levels $\nt=1$ and $\nt=2$, which is close to true optimal value of  $\nu^*=0.07$.

\begin{figure}[hbt!]
  \begin{subfigure}[b]{0.5\linewidth}
    \centering
  \centering
  \includegraphics[scale=0.42,trim={0cm 8cm 0cm 7cm}]{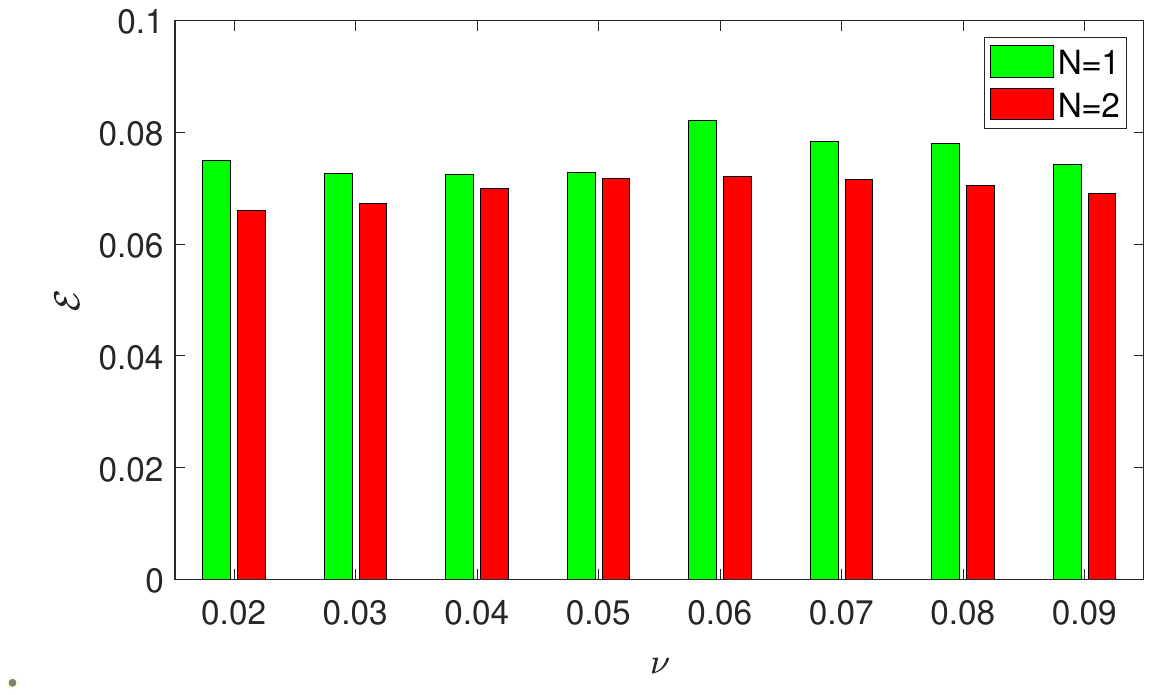}
  \caption{}\label{fig:compN1N2}
\end{subfigure}
\begin{subfigure}[b]{0.5\linewidth}
  \centering
  \includegraphics[scale=0.42,trim={0cm 8cm 0cm 7cm}]{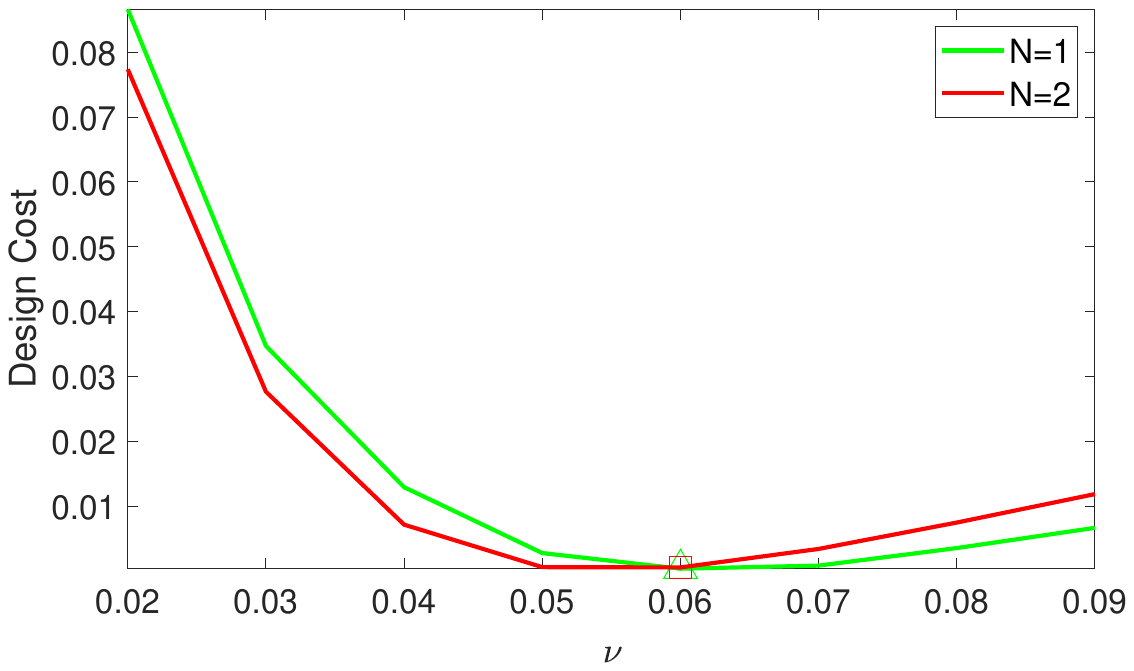}
    \caption{}
\end{subfigure}
\caption{(a) CL+VQLS solution error $\mathcal{E}$ (\ref{eq:err}) for the inverse Burgers problem for different values of $\nu$ and two CL truncation levels $\nt=1$ and $\nt=2$. (b)Design cost (Eq.~\ref{eq:opttrans}) as a function of $\nu$ obtained via our nBVQPCO framework for CL truncation levels $\nt=1$ and $\nt=2$. Also, marked are $\nu$ values where the cost takes the smallest value, which is $\nu=0.06$ for both the truncation levels. The true optima is known to be $\nu=0.07$.}\label{fig:compCostN1N2}
\end{figure}

\section{Discussion}
In this section we discuss various pros/cons of the nBVQPCO framework and a path towards fully fault tolerant quantum computing implementation. While nBVQPCO framework is well suited for NISQ implementation, like other variational quantum algorithms it is difficult to subject it to a fully theoretical complexity analysis. As pointed out in the Remark \ref{rem:vqlsscaling}, it will be worthwhile to analyze VQLS empirical scaling behavior for sparse and structured matrices.  Furthermore, the error and computational complexity results we presented depend on the condition~\Cref{eq:R2}, which can be conservative and too restrictive to be satisfied in practical applications. As pointed out in the Introduction, for CFD type applications, one can work with alternative form of conservation laws, e.g. LBM \cite{li2023potential}, for which the condition~\Cref{eq:R2} can more readily be met in practical regimes of interest.  

One of the bottlenecks in the nBVQPCO framework is the computation of VQLS cost function which scales polynomially with the number of LCU terms. By replacing conventionally employed Pauli basis with sigma basis, we showed in our application that the number of LCU terms scales more favorably. Ideally, one would like this scaling to have a polylogarithmic dependence on the matrix size, and further enhancing the sigma basis type approaches would be necessary to make the the proposed framework scalable. Other variational quantum algorithms such as variational quantum eigensolver (VQE), and quantum approximate optimization algorithm (QAOA) can also significantly benefit from such efficient LCU decompositions.  

In the nBVQPCO framework one can replace VQLS by the quantum linear systems algorithm (QLSA) \cite{dervovic2018quantum,harrow2009quantum,costa2022optimal}. This could be beneficial since QLSA has provable exponential advantage over classical linear system solvers and thus comes with rigorous guarantees unlike variational methods. However,  QLSA based implementation is expected to have high quantum resource needs. For instance in \cite{penuel2024feasibility}, the authors obtained a detailed quantum resource estimate for implementing CL based LBM simulation of incompressible flow fields and concluded that QLSA based framework can only be feasible in a fault tolerant quantum computing setting. This study also highlighted that dominant quantum resource need arise due to the block encoding step which is required to implement the oracle in QLSA. Block encoding in fault tolerant implementation can be thought of as an analogous step to the LCU decomposition in NISQ implementation. Techniques like sigma basis which exploit sparsity and structure for more efficient LCU decomposition can also be potentially leveraged to make block encoding schemes more efficient. 

\section{Conclusions}\label{sec:conc}
In this paper we presented nBVQPCO, a novel variational quantum framework for nonlinear PDE constrained optimization problems. The proposed framework utilizes Carleman linearization, VQLS algorithm and a black box optimizer nested in bi-level optimization structure. We presented a detailed computational error and complexity analysis to establish potential benefits of our framework over classical techniques under certain assumptions. We demonstrated the framework on a inverse problem and presented simulation results. The analysis and results demonstrate the correctness of our framework for solving nonlinear simulation-based design optimization problems. 

Future work will involve studying and mitigating the effect of device/ measurement noise and making the nBVQPCO framework robust. It will also be important to study the scalability of the framework by applying it to larger problem sizes and implementing on quantum hardware. Finally, exploring and refining the framework to better exploit sparsity/structure and extension to fault tolerant setting as discussed above are also avenues for future research.


\section*{Acknowledgements}
This research was developed with funding from the Defense Advanced Research Projects Agency (DARPA). The views, opinions, and/or findings expressed are those of the author(s) and should not be interpreted as representing the official views or policies of the Department of Defense or the U.S. Government.

\appendix
\section*{Appendix}

\section{Proof of~\Cref{thm:carl}}\label{sec:proof_thmcarl}
\carl*
\begin{proof}
Let $\wvh_d=(\wvh^\tra(0),\cdots,\wvh^\tra(\tsteps h))^\tra$. Using the Lemma \ref{lem:euler}
\begin{eqnarray}
\|\wvh_d-\tilde{\wv}\|^2&=& \sum_{k=0}^{\tsteps} \|\wvh(kh)-\wvh^k\|^2, \notag \\
  &\leq& \left(3 \nt^{2.5} h^2\left(\left(\|\F_2\|+\|\F_1\|+\|\F_0\|\right)^2+\|\dot{\F}_0\|\right)\right)^2\sum_{k=0}^{\tsteps}k^2,
\end{eqnarray}
which implies
\begin{eqnarray}
\|\wvh_d-\tilde{\wv}\|&\leq & 3 \nt^{2.5} h^2\left((\|\F_2\|+\|\F_1\|+\|\F_0\|)^2+\|\dot{\F}_0\|\right)\tsteps^{3/2},\notag\\
  &\leq & 3 \nt^{2.5} h^{1/2}\left((\|\F_2\|+\|\F_1\|+\|\F_0\|)^2+\|\dot{\F}_0\|\right)T^{3/2}.
\end{eqnarray}

Similarly, using the Lemma \ref{lem:cltruc}, we get
\begin{eqnarray}
\|\wv_c-\wvh_d\|&\leq& \nt \|\F_2\| \|\uv_{0}\|^{\nt+1} h\sqrt{\sum_{k=0}^{\tsteps}k^2}, \notag \\
  &\leq& \nt \|\F_2\| \|\uv_{0}\|^{\nt+1} T^{3/2}h^{-1/2}.
\end{eqnarray}

We next show that for given $\epsilon^\prime>0$, one can choose $\nt$ and $h$ such that they satisfy a system of inequalities
\begin{eqnarray}
 && 3\nt^{2.5} h^{1/2}\left((\|\F_2\|+\|\F_1\|+\|\F_0\|)^2+\|\dot{\F}_0\|\right)T^{3/2} \leq  \epsilon^\prime/2, \label{eq:ineq1} \\
 && \nt \|\F_2\| \|\uv_{0}\|^{\nt+1} T^{3/2}h^{-1/2} \leq  \epsilon^\prime/2, \label{eq:ineq2}\\
 && \|\etav(t)\|\leq t\nt \|\F_2\| \|\uv_{0}\|^{\nt+1} \leq T\nt \|\F_2\| \|\uv_{0}\|^{\nt+1} \leq  \frac{\|\uv(T)\|}{4},\label{eq:ineq3}\\
 && h\leq \min \left\{\frac{1}{\nt\|\F_1\|},\frac{2(|Re(\lambda_1)|-\|\F_2\|-\|\F_0\|)}{\nt(|Re(\lambda_1)|^2-(\|\F_1\|+\|\F_0\|)^2+\|\F_1\|^2)}\right\},\label{eq:ineq4}
\end{eqnarray}
where, the last two inequalities comes form the Assumption (\ref{eq:etabound}) and the requirement (\ref{eq:h}) in the Lemma \ref{lem:euler}. Combining (\ref{eq:ineq1}) and (\ref{eq:ineq4}) implies
\begin{equation}\label{eq:href}
h\leq \left\{\frac{1}{\nt\|\F_1\|},\frac{2(|Re(\lambda_1)|-\|\F_2\|-\|\F_0\|)}{\nt(|Re(\lambda_1)|^2-(\|\F_1\|+\|\F_0\|)^2+\|\F_1\|^2)},\frac{(\epsilon^\prime)^2}{36\,\nt^5T^3\left(2(\|\F_2\|+\|\F_1\|+\|\F_0\|)^2+\|\dot{\F}_0\|\right)^2}\right\}.
\end{equation}
Consequently, we can take $h$ to be of the form $h=K/N^\alpha$, where $\alpha>1$ and $K>0$ is an appropriate constant. Thus, (\ref{eq:ineq2}) implies
\begin{equation*}
\nt \|\F_2\| \|\uv_{0}\|^{\nt+1} T^{3/2}\leq \frac{1}{2}\epsilon^\prime\sqrt{h}\leq \frac{\epsilon^\prime\sqrt{K}}{2 N^{\alpha/2}},
\end{equation*}
and $\nt$ needs to satisfy an equation of the form
\begin{equation}
\nt^{1+\alpha/2}\|\uv_{0}\|^{\nt}\leq \delta_1=\frac{\epsilon^\prime\sqrt{K}}{2T^{3/2} \|\F_2\|\|\uv_{0}\| },\label{eq:nteq1}
\end{equation}
and similarly from (\ref{eq:ineq3})
\begin{equation}
\nt\|\uv_{0}\|^{\nt}\leq \delta_2=\frac{\|\uv(T)\|}{4 T \|\F_2\|\|\uv_{0}\|}.\label{eq:nteq2}
\end{equation}
Note that if we choose $\nt>1$ such that
\begin{equation}
\nt^{1+\alpha/2}\|\uv_{0}\|^{\nt}\leq \delta=\min\{\delta_1,\delta_2\},\label{eq:nteq3}
\end{equation}
then the inequality (\ref{eq:nteq2}) is satisfied. Recall, from the Remark \ref{rem1} we can always assume $\|\uv_{0}\|<1$.  Consider a function of the form
\begin{equation*}
f(x)=x^\beta a^x,
\end{equation*}
where, $0<a<q, \beta\geq 1$ and $x\ge 0$. Then it is easy to show that $f(x)$ is monotonically decreasing for $x>\frac{\beta}{\log \frac{1}{a}}$ and $\lim_{x\rightarrow \infty} f(x)=0$. Thus,  one can always find an integer 
\begin{equation}\label{eq:ntref}
\nt>\frac{1+\alpha/2 }{\log \frac{1}{\|\uv_{0}\|}}
\end{equation}
such that (\ref{eq:nteq3}) is satisfied.

Thus, we conclude that for given $\epsilon^\prime$ one can always choose $\nt$ and $h$, such that
\begin{equation*}
\|\wv_c-\wvh_d\|\leq \epsilon^\prime/2, \quad \|\wvh_d-\tilde{\wv}\|\leq \epsilon^\prime/2,
\end{equation*}
and so
\begin{equation}\label{eq:err1p}
\|\wv_c-\tilde{\wv}\|\leq \|\wv_c-\wvh_d\|+\|\wvh_d-\tilde{\wv}\|\leq \epsilon^\prime.
\end{equation}
Finally, note that
\begin{eqnarray}\label{eq:normstep}
&&\left\|\frac{\wv_c}{\|\wv_c\|}-\frac{\tilde{\wv}}{\|\tilde{\wv}\|}\right\|\leq \left\|\frac{\wv_c}{\|\wv_c\|}-\frac{\tilde{\wv}}{\|\wv_c\|}+\frac{\tilde{\wv}}{\|\wv_c\|}-\frac{\tilde{\wv}}{\|\tilde{\wv}\|}\right\|,\notag\\
&\leq & \frac{\|\wv_c-\tilde{\wv}\|}{\|\wv_c\|}+\frac{\|\|\wv_c\|-\|\tilde{\wv}\|\|}{\|\wv_c\|}\leq \frac{2\|\wv_c-\tilde{\wv}\|}{\|\wv_c\|}\leq \frac{2\epsilon^\prime}{\|\wv_c\|}, 
\end{eqnarray}
where we have use the fact
\begin{equation*}
\frac{\|\|\wv_c\|-\|\tilde{\wv}\|\|}{\|\wv_c\|}\leq \frac{\|\wv_c-\tilde{\wv}\|}{\|\wv_c\|}.
\end{equation*}
Thus, for given $\epsilon$, by choosing $\epsilon^\prime=\|\wv_c\|\epsilon/2$ and determining $\nt$ and $h$ as described above, we can conclude,
\begin{equation*}
\|\frac{\wv_c}{\|\wv_c\|}-\frac{\tilde{\wv}}{\|\tilde{\wv}\|}\|\leq \epsilon.
\end{equation*}
as required.
\end{proof}

\section{Proof of~\Cref{thm:classeuler}}\label{sec:proof_thmclasseuler}

\classeuler*

\begin{proof}
Consider the R.H.S. of (\ref{eq:qdc})
\begin{equation}
\F(t,\uv)=\F_0(t)+\F_1\uv+\F_2\uv^{[2]}.
\end{equation}
Let $r>0$, then for any $\uv_1,\uv_2\in \{\uv:\|\uv\|\leq r\}$ 
\begin{eqnarray}
\|\F(t,\uv_1)-\F(t,\uv_2)\|&\leq& \|\F_1\||\ \uv_1-\uv_2\|+\|\F_2\||\ \uv_1^{[2]} -\ \uv_2^{[2]} \|\notag\\
&=&\|\F_1\||\ \uv_1-\uv_2\|+\|\F_2\||\ \uv_1\otimes\uv_1-\uv_1\otimes\uv_2+\uv_1\otimes\uv_2 -\ \uv_2^{[2]} \|\notag\\
&\leq&\|\F_1\||\ \uv_1-\uv_2\|+\|\F_2\| |\ \uv_1\| \|\uv_1-\uv_2\|+\|\F_2\| |\ \uv_2\| \|\uv_1-\uv_2\|\notag\\
&\leq&\|\F_1\| |\ \uv_1-\uv_2\|+(\|\F_2\| |\ \uv_1\| +|\uv_2\|)\|\uv_1-\uv_2\|\notag.
\end{eqnarray}
Thus
\begin{equation}\label{eq:lip}
\|\F(t,\uv_1)-\F(t,\uv_2)\|\leq L \|\uv_1-\uv_2\|
\end{equation}
where
\begin{equation}\label{eq:lipcst}
L=\|\F_1\|+2r\|\F_2\|.
\end{equation}

As shown in the Appendix A.1  in \cite{liu2021efficient}, under the Assumptions \ref{assum1} the solution of (\ref{eq:qdc}) is bounded, i.e.
\begin{equation}\label{eq:solboundcnt}
\|\uv(t)\|\leq \|\uv_0\|, \forall t\geq 0.
\end{equation}
Hence, taking the time derivative of (\ref{eq:qdc}) and since $\|\dot{\F}_0\|$ is bounded as per the Assumptions \ref{assum1}, we conclude that $\|\ddot{\uv}(t)\|$ is bounded.  Then by Taylor expansion
\begin{equation}\label{eq:loctrunc}
\uv((k+1)h)=\uv(kh)+h\F(kh,\uv(kh))+\tauv_k,
\end{equation}
where $\tauv_k$ is the local truncation error which can be bounded as
\begin{equation}
\|\tauv_k\|\leq\frac{M_k}{2}h^2,\label{eq:taubound}
\end{equation}
with $M_k=\max_{t\in [kh,(k+1)h]}\|\ddot{\uv}(t)\|< \infty$. Let $M=\max_k M_k$. 

Similar to (\ref{eq:solboundcnt}) we next show that for sufficiently small $h$, iterates of~\Cref{eq:direuler} also satisfy
\begin{equation}\label{eq:solbounddis}
\|\uv^k\|\leq \|\uv^0\|, \forall k\geq 0,
\end{equation}
and thus the Euler scheme is stable. To show this, from~\Cref{eq:direuler} we get
\begin{eqnarray*}
\frac{(\uvh^{k+1})^\tra\uvh^{k+1}-(\uvh^{k})^\tra\uvh^{k}}{h}&=&(\uvh^{k})^\tra\F_0(kh)+\F_0(kh)^\tra\uvh^{k}+(\uvh^k)^\tra(\F_1+\F_1^\tra)\uvh^k+\\
&+&(\uvh^{k})^\tra\F_2(\uvh^k)^{[2]}+((\uvh^k)^{[2]})^\tra\F_2^\tra\uvh^{k}+h\|\F(kh,\uvh^{k})\|^2,
\end{eqnarray*}
which implies
\begin{eqnarray*}
\frac{\|\uvh^{k+1}\|^2-\|\uvh^{k}\|^2}{h}&\leq&2\|\F_0\|\|\uvh^{k}\|+2\mbox{Re}(\lambda_1)\|\uvh^k\|^2+2\|\F_2\|\| \uvh^{k}\|^3+h\|\F(kh,\uvh^{k})\|^2.
\end{eqnarray*}


Under the condition (\ref{eq:R2}), for $k=0$, since 
\begin{equation*}
\|\F_0\|+\mbox{Re}(\lambda_1)\|\uvh^0\|+\|\F_2\|\| \uvh^{0}\|^2< 0,
\end{equation*}
one can choose $h$ 
\begin{equation}\label{eq:h0}
h\leq h_0=-\frac{\|\F_0\|+\mbox{Re}(\lambda_1)\|\uvh^0\|+\|\F_2\|\| \uvh^{0}\|^2}{\|\F(kh,\uvh^{0})\|^2}\|\uvh^0\|,
\end{equation}
so that
\begin{equation*}
\|\uvh^{1}\|\leq \|\uvh^{0}\|.
\end{equation*}
Following similar argument one can show that $\|\uvh^{2}\|\leq \|\uvh^{1}\|$ and so on, concluding (\ref{eq:solbounddis}).

Next, subtracting (\ref{eq:direuler}) from (\ref{eq:loctrunc}) leads to 
\begin{equation*}
\uv((k+1)h)-\uvh^{k+1}=\uv(kh)+h\F(kh,\uv(kh))+\tauv_k-(\uvh^{k}+h\F(kh,\uvh^k)),
\end{equation*}
which implies that the norm of the error $e_{k}=\|\uv(kh)-\uvh^{k}\|$ satisfies
\begin{equation}
e_{k+1}\leq (1+hL)e_k+\|\tauv_k\|,
\end{equation}
where we have used the facts (\ref{eq:solboundcnt}), (\ref{eq:solbounddis}) and (\ref{eq:lip}) with $L=\|\F_1\|+2\|\F_2\|\|\uv_0\|$.  Iterating on $k$ we get
\begin{equation}
e_{k+1}\leq (1+hL)^{k+1}e_0+((1+hL)^{k}\|\tauv_0\|+(1+hL)^{k-2}\|\tauv_2\|+\cdots \|\tauv_k\|), 
\end{equation}
or 
\begin{equation}
e_{k+1}\leq (1+hL)^{k+1}e_0+\max \|\tauv_k\|\frac{(1+hL)^k+1}{hL}, 
\end{equation}
Since $e_0=0$ for Euler method, using (\ref{eq:taubound}) and $1+Lh\leq e^{hL}$ we conclude that
\begin{equation*}
\|\uv(kh)-\uvh^{k}\|=e_{k}\leq C h.
\end{equation*}
where $C=\frac{M}{2L}(e^{LT}-1)$. It then follows that
\begin{eqnarray}
\|\uv_c-\tilde{\uv}\|^2&=& \sum_{k=0}^{\tsteps} \|\uv(kh)-\uvh^k\|^2\leq (Ch)^2\tsteps,\notag
\end{eqnarray}
which implies
\begin{eqnarray}
\|\uv_c-\tilde{\uv}\|&\leq &  Ch\tsteps^{1/2}\leq C h^{1/2} T^{1/2}.
\end{eqnarray}
Thus, for any given $\epsilon^\prime$ one can choose 
\begin{equation}\label{eq:hebound}
h\leq \min\Big\{\frac{(\epsilon^\prime)^2}{C^2T},h_0\Big\},
\end{equation}
where $h_0$ is as defined in (\ref{eq:h0}),  so that
\begin{equation*}
\|\uv_c-\tilde{\uv}\|\leq \epsilon^\prime.
\end{equation*}
Finally, for given $\epsilon$, by choosing $\epsilon^\prime=\|\uv_c\|\epsilon/2$, selecting $h$ as described above and following similar step (\ref{eq:normstep}) as discussed in the proof of~\Cref{thm:carl}, we conclude
\begin{equation*}
\left\|\frac{\uv_c}{\|\uv_c\|}-\frac{\tilde{\uv}}{\|\tilde{\uv}\|}\right\|\leq \epsilon.
\end{equation*}
\end{proof}

\bibliographystyle{plain}
\bibliography{references}

\end{document}